\documentclass[acmsmall,screen]{acmart}
\acmJournal{PACMPL}
\acmVolume{1}
\acmNumber{POPL} 
\acmArticle{1}
\acmYear{2021}
\acmMonth{1}
\acmDOI{} 
\setcopyright{acmlicensed}
\acmPrice{}
\acmDOI{10.1145/3434311}
\acmYear{2021}
\copyrightyear{2021}
\acmSubmissionID{popl21main-p195-p}
\acmJournal{PACMPL}
\acmVolume{5}
\acmNumber{POPL}
\acmArticle{30}
\acmMonth{1}


\usepackage{graphicx}
\usepackage{adjustbox}
\usepackage[T1]{fontenc}
\usepackage{stmaryrd}
\usepackage{proof}
\usepackage{xspace}
\usepackage{multirow}
\usepackage{pifont} 
\usepackage{arydshln}
\usepackage{enumitem}
\usepackage{subcaption}

\usepackage{tikz}
\usepackage[framemethod=tikz]{mdframed}
\usepackage{xcolor}
\usepackage{wrapfig}
\usepackage{pgfplots}
\usepackage{mathtools}
\usepackage{float}
\usepackage[leftbars, color]{changebar}

\usepackage{microtype}

\newif\ifdispComments
\dispCommentsfalse

\ifdispComments
\else
\nochangebars
\fi

\newtheorem*{theorem*}{Theorem}

\newcommand{\ball}[1]{\hspace{.7em}\tikz[baseline=(myanchor.base)] \node[circle,fill=.,inner sep=1pt,minimum size=1pt] (myanchor) {\color{-.}\bfseries\footnotesize #1};\hspace{.5em}}
\newcommand{\ballns}[1]{\tikz[baseline=(myanchor.base)] \node[scale=0.9,circle,fill=.,inner sep=1pt,minimum size=1pt] (myanchor) {\color{-.}\bfseries\footnotesize #1};}
\newcommand{\sball}[1]{\hspace{.7em}\tikz[baseline=(myanchor.base)] \node[scale=0.75,circle,fill=.,inner sep=1pt,minimum size=1pt] (myanchor) {\color{-.}\bfseries\footnotesize #1};\hspace{.5em}}
\newcommand{\sballns}[1]{\hspace{.7em}\tikz[baseline=(myanchor.base)] \node[scale=0.75,circle,fill=.,inner sep=1pt,minimum size=1pt] (myanchor) {\color{-.}\bfseries\footnotesize #1};}
\newcommand{\sballlrns}[1]{\tikz[baseline=(myanchor.base)] \node[scale=0.75,circle,fill=.,inner sep=1pt,minimum size=1pt] (myanchor) {\color{-.}\bfseries\footnotesize #1};}

\makeatletter
\let\save@mathaccent\mathaccent
\newcommand*\if@single[3]{%
  \setbox0\hbox{${\mathaccent"0362{#1}}^H$}%
  \setbox2\hbox{${\mathaccent"0362{\kern0pt#1}}^H$}%
  \ifdim\ht0=\ht2 #3\else #2\fi
  }
\newcommand*\rel@kern[1]{\kern#1\dimexpr\macc@kerna}
\newcommand*\widebar[1]{\@ifnextchar^{{\wide@bar{#1}{0}}}{\wide@bar{#1}{1}}}
\newcommand*\wide@bar[2]{\if@single{#1}{\wide@bar@{#1}{#2}{1}}{\wide@bar@{#1}{#2}{2}}}
\newcommand*\wide@bar@[3]{%
  \begingroup
  \def\mathaccent##1##2{%
    \let\mathaccent\save@mathaccent
    \if#32 \let\macc@nucleus\first@char \fi
    \setbox\z@\hbox{$\macc@style{\macc@nucleus}_{}$}%
    \setbox\tw@\hbox{$\macc@style{\macc@nucleus}{}_{}$}%
    \dimen@\wd\tw@
    \advance\dimen@-\wd\z@
    \divide\dimen@ 3
    \@tempdima\wd\tw@
    \advance\@tempdima-\scriptspace
    \divide\@tempdima 10
    \advance\dimen@-\@tempdima
    \ifdim\dimen@>\z@ \dimen@0pt\fi
    \rel@kern{0.6}\kern-\dimen@
    \if#31
      \overline{\rel@kern{-0.6}\kern\dimen@\macc@nucleus\rel@kern{0.4}\kern\dimen@}%
      \advance\dimen@0.4\dimexpr\macc@kerna
      \let\final@kern#2%
      \ifdim\dimen@<\z@ \let\final@kern1\fi
      \if\final@kern1 \kern-\dimen@\fi
    \else
      \overline{\rel@kern{-0.6}\kern\dimen@#1}%
    \fi
  }%
  \macc@depth\@ne
  \let\math@bgroup\@empty \let\math@egroup\macc@set@skewchar
  \mathsurround\z@ \frozen@everymath{\mathgroup\macc@group\relax}%
  \macc@set@skewchar\relax
  \let\mathaccentV\macc@nested@a
  \if#31
    \macc@nested@a\relax111{#1}%
  \else
    \def\gobble@till@marker##1\endmarker{}%
    \futurelet\first@char\gobble@till@marker#1\endmarker
    \ifcat\noexpand\first@char A\else
      \def\first@char{}%
    \fi
    \macc@nested@a\relax111{\first@char}%
  \fi
  \endgroup
}
\makeatother

\usepackage[leftbars,color]{changebar}
\cbcolor{red}

\newcommand{\rone}{(\emph{i})~}
\newcommand{\rtwo}{(\emph{ii})~}

\newcommand{\ex}{\mathsf{e}}
\newcommand{\examples}{\mathsf{E}}
\newcommand{\xmark}{\ding{55}}%
\def\namenay{\textsc{nay}\xspace}

\def\nope{\textsc{nope}\xspace}

\usepackage{color}
\ifdispComments
  \newcommand {\loris}[1]{{\color{blue}\bf{L: #1}\normalfont}}
  \newcommand {\heping}[1]{{\color{red}\bf{H: #1}\normalfont}}
  \newcommand {\tom}[1]{{\color{green}\bf{T: #1}\normalfont}}
  \newcommand {\jinwoo}[1]{{\color{violet}\bf{J: #1}\normalfont}}
\else
  \newcommand {\loris}[1]{}
  \newcommand {\heping}[1]{}
  \newcommand {\tom}[1]{}
  \newcommand {\jinwoo}[1]{}
\fi

\DeclareRobustCommand\ourtool{\textsc{MESSY}\xspace}
\newcommand{\name}{\ourtool}

\def\nope{\textsc{Nope}\xspace}
\def\nay{\textsc{Nay}\xspace}

\def\semgus{\textsc{SemGuS}\xspace}
\def\sygus{\textsc{SyGuS}\xspace}

\def\nat{\mathbb{N}}
\newcommand{\type}{\tau}
\newcommand{\sem}[1]{\llbracket{#1}\rrbracket}
\newcommand{\Etrue}{{{\mathsf{true}}}}
\newcommand{\Efalse}{{{\mathsf{false}}}}
\newcommand{\Eifthenelse}[3]{{\mathsf{if}}\ {#1}\ {\mathsf{then}}\ {#2}\ {\mathsf{else}}\ {#3}}
\newcommand{\Eseq}[2]{{#1}{\mathsf{;}}\ {#2}}
\newcommand{\Earrayaccess}[2]{{#1}[#2]}

\newcommand{\Eassign}[2]{{{#1}\ \mathsf{:=}\ {#2}}}
\newcommand{\Earrayupdate}[3]{{{#1}[{#2}]\ \mathsf{:=}\ {#3}}}

\newcommand{\Vect}[1]{\vec {#1}}
\newcommand{\Vproj}[2]{\textsc{proj}({#1}, {#2})}
\newcommand{\Vjoin}[2]{\textsc{merge}({#1}, {#2})}
\newcommand{\Ereduce}[3]{\sem{#1}(#2,#3)}
\newcommand{\Sreduce}[3]{\sem{#1}(#2,#3)}
\newcommand{\eEreduce}[3]{\sem{#1}_\examples(#2,#3)}
\newcommand{\eSreduce}[3]{\sem{#1}_\examples(#2,#3)}
\newcommand{\Ewhile}[2]{\mathsf{while} \; {#1} \; \mathsf{do} \; {#2}}

\newcommand{\aL}{\mathcal{L}}

\newcommand{\VG}{\Vect{\Gamma}}

\newcommand{\Vv}{\Vect{v}}

\newcommand{\rWT}{\mathsf{WTrue}}

\newcommand{\aI}{\mathcal{L}}

\newcommand{\aT}{\mathsf{T}}

\newcommand{\syntax}{\texttt{syn}}
\newcommand{\semant}{\texttt{sem}}

\newcommand{\Frreduce}[3]{\syntax_{#3}(#1,#2)}

\newcommand{\BFrreduce}[2]{\Frreduce{#1}{#2}{B}}
\newcommand{\SFrreduce}[2]{\Frreduce{#1}{#2}{S}}

\newcommand{\Startrreduce}[2]{\Frreduce{#1}{#2}{\mathit{Start}}}

\newcommand{\FrreduceTree}[2]{\syntax_{#2}(#1)}
\newcommand{\EFrreduceTree}[1]{\FrreduceTree{#1}{E}}
\newcommand{\BFrreduceTree}[1]{\FrreduceTree{#1}{B}}
\newcommand{\SFrreduceTree}[1]{\FrreduceTree{#1}{S}}
\newcommand{\NFrreduceTree}[1]{\FrreduceTree{#1}{N}}
\newcommand{\StartrreduceTree}[1]{\FrreduceTree{#1}{\mathit{Start}}}

\newcommand{\FrreduceFused}[3]{\syntax^{\mathsf{fused}}_{#3}(#1,#2)}
\newcommand{\BFrreduceFused}[2]{\FrreduceFused{#1}{#2}{B}}
\newcommand{\SFrreduceFused}[2]{\FrreduceFused{#1}{#2}{S}}
\newcommand{\NFrreduceFused}[2]{\FrreduceFused{#1}{#2}{N}}
\newcommand{\EFrreduceFused}[2]{\FrreduceFused{#1}{#2}{E}}

\newcommand{\NFnrreduceFused}[3]{\FrreduceFused{#1}{#2}{N_{#3}}}

\newcommand{\tNode}[1]{\mathsf{}{#1}}
\newcommand{\uTree}[2]{\mathsf{Tree}_{\tNode{#1}}({#2})}
\newcommand{\bTree}[3]{\mathsf{Tree}_{\tNode{#1}}({#2}, {#3})}

\newcommand{\NFnrreduceTree}[2]{\FrreduceTree{#2}{N_{#1}}}

\newcommand{\Flreduce}[3]{\semant_{#3}(\langle #1\rangle,\langle #2\rangle)}
\newcommand{\EFlreduce}[2]{\Flreduce{#1}{#2}{E}}
\newcommand{\BFlreduce}[2]{\Flreduce{#1}{#2}{B}}
\newcommand{\SFlreduce}[2]{\Flreduce{#1}{#2}{S}}
\newcommand{\NFlreduce}[2]{\Flreduce{#1}{#2}{N}}

\newcommand{\NFnlreduce}[3]{\Flreduce{#2}{#3}{N_{#1}}}
\newcommand{\Startlreduce}[2]{\Flreduce{#1}{#2}{\mathit{Start}}}

\newcommand{\FlreduceTree}[3]{\semant_{#3}(\langle #1\rangle,#2)}

\newcommand{\BFlreduceTree}[2]{\FlreduceTree{#1}{#2}{B}}
\newcommand{\SFlreduceTree}[2]{\FlreduceTree{#1}{#2}{S}}
\newcommand{\NFlreduceTree}[2]{\FlreduceTree{#1}{#2}{N}}
\newcommand{\RFlreduceTree}[2]{\FlreduceTree{#1}{#2}{R}}

\newcommand{\NFnlreduceTree}[3]{\FlreduceTree{#2}{#3}{N_{#1}}}
\newcommand{\StartlreduceTree}[2]{\FlreduceTree{#1}{#2}{\mathit{Start}}}

\newcommand{\rk}{\textit{rk}}

\newcommand{\Omit}[1]{}

\newcommand{\SemStart}{\mathsf{sem}_{\mathit{Start}}}
\def\q#1{\textbf{Q{#1}}}
\newcommand{\mypar}[1]{\vspace{1mm} \emph{{#1.} }}

\newcommand{\defeq}{\triangleq}
\usepackage[breakable]{tcolorbox}
\newenvironment{mybox}[1][gray!20]{
	\begin{tcolorbox}[   
		breakable,
		left=0pt,
		right=0pt,
		top=0pt,
		bottom=-1pt,
		colback=#1,
		colframe=#1,
		width=\dimexpr\textwidth\relax,
		boxsep=4pt,
		arc=0pt,outer arc=0pt,
		]
	}{
\end{tcolorbox}
}

\newcounter{resq}

\title[Semantics-Guided Synthesis]{Semantics-Guided Synthesis}
 \author{Jinwoo Kim}
  \affiliation{
    \institution{Unversity of Wisconsin-Madison} \country{USA}} \email{pl@cs.wisc.edu}
  \author{Qinheping Hu}
  \affiliation{
    \institution{Unversity of Wisconsin-Madison} \country{USA}} \email{qhu28@wisc.edu}
  \author{Loris D'Antoni}
  \affiliation{
    \institution{Unversity of Wisconsin-Madison} \country{USA}} \email{loris@cs.wisc.edu}
  \author{Thomas Reps}
  \affiliation{
    \institution{Unversity of Wisconsin-Madison} \country{USA}} \email{reps@cs.wisc.edu}


\begin{document}

\begin{abstract}
  This paper develops a new framework for program synthesis, called \emph{semantics-guided synthesis} 
  (\semgus), that allows a user to provide both the syntax and the semantics for the constructs in the language.
  \semgus accepts a recursively defined big-step semantics, which allows it, for example, to be used to 
  specify and solve synthesis problems over an imperative programming language that may contain 
  loops with unbounded behavior.
  The customizable nature of \semgus also allows synthesis problems to be
  defined over a non-standard semantics, such as an abstract semantics.
  In addition to the \semgus framework, we develop an algorithm for solving \semgus problems 
  that is capable of both synthesizing programs and proving unrealizability, by encoding a \semgus 
  problem as a proof search over Constrained Horn Clauses:
  in particular, our approach is the first that we are aware of that can prove unrealizabilty for 
  synthesis problems that involve
  imperative programs with unbounded loops, over an infinite syntactic search space.
  We implemented the technique in a tool called \ourtool, and applied it
  to \sygus problems (i.e., over expressions), synthesis problems over
  an imperative programming language, and synthesis problems over regular expressions.
\end{abstract}


  \maketitle

\section{Introduction}
\label{Se:Introduction}

Program synthesis refers to the task of finding a program 
within a given search space that meets a given behavioral specification
(typically a logical formula or
a set of input-output examples).
Program synthesis has been studied from a variety of perspectives, which have led to
great practical advances in specific domains~\cite{swizzle, flashfill, lambda2}.

The proliferation of domain-specific synthesis tools has led to numerous attempts to build  
frameworks that allow one to define and solve synthesis problems in a general fashion.
Tools such as Sketch~\cite{sketch} and
Rosette~\cite{rosette}
have introduced the notion of a solver-aided language, which 
allows one to define a synthesis problem using a specialized language and
then solve the specified problem using a constraint solver. To retain the ability
to solve practical problems, these tools have restricted their languages 
in ways that enable the use of constraint-based synthesis methods---e.g., Sketch and Rosette do not allow arbitrary search spaces involving
programs of unbounded size.

While solver-aided languages made synthesis more ``programmable'', their mutual incompatibility
and language restrictions led to a natural questions:
\emph{Can we define synthesis problems in a language-agnostic way?}
This question was partly answered by the framework of
of \emph{syntax-guided synthesis} (\sygus)~\cite{sygus}, 
which provides a logical framework for defining synthesis problems.
In a \sygus problem, the search space is described using a context-free grammar
of terms from a given theory, and the behavioral specification is expressed using a
formula in that same logical theory.
The unified logical format offered by \sygus spurred researchers to design
synthesizers that could solve problems defined in the \sygus format~\cite{alur2017esolver,cvc4}, and
these solvers compete annually in \sygus competitions~\cite{sygus-comp}.
However, \sygus introduced its own limitation: namely, 
that the semantics of \sygus problems are limited to those from a 
fixed theory, such as linear integer 
arithmetic (LIA) or bitvectors.
This limitation has created a gap between the two approaches: solver-aided languages
are unable to express \sygus problems with infinite search spaces, 
while \sygus cannot express problems with semantics outside of a supported 
theory, such as imperative programs containing loops (which could be modeled 
using tools like Sketch and Rosette).

\mypar{The \semgus Framework} In this paper, we bridge this gap and present a new synthesis framework, called 
\emph{semantics-guided synthesis} (\semgus), that attempts to encompass and generalize 
the two approaches.
Like \sygus, the goal of \semgus is to provide a \emph{general,
logical framework} that expresses the \emph{core computational
problem} of program synthesis~\cite{sygus}, without being tied to a
specific solution or implementation.
However, in addition to a syntactic search space and a behavioral
specification, \semgus also allows the user to define the
\emph{semantics} of constructs in the grammar in terms of a set of
inference rules---hence the name ``semantics-guided synthesis''.

\cbstart
\jinwoo{Shepherd: Explanation of what we mean by a synthesis framework.}
By a \emph{framework}, we mean the conceptual underpinnings
that allows one to build a tool to automate the creation of
solutions for problems in some domain.
The canonical example is the theory of parsing, which provides
the underpinnings of the tool \texttt{yacc} \cite{BellLabs-TR32:J75},
which automates the construction of parsers.

For example, consider the problem that \texttt{yacc} addresses.

\begin{mdframed}[innerleftmargin = 3pt, innerrightmargin = 3pt, skipbelow=-0.25em]
\begin{itemize}[labelsep=*, leftmargin=*]
  \item
    An instance of a parsing problem, Parse($L$,$s$), has two parameters:
    $L$, a context-free language; and $s$, a string to be parsed.
    String $s$ changes more frequently than language $L$.
  \item
    Context-free grammars are a formalism for specifying
    context-free languages.
  \item
    Create a tool that implements the following specification:
    \begin{itemize}
      \item
        Input: a context-free grammar that describes language $L$.
      \item
        Output: a parser, \texttt{yyparse()}, such that invoking
        \texttt{yyparse()} on $s$ computes Parse($L$,$s$).
    \end{itemize}
\end{itemize}
\end{mdframed}

\noindent
Thus, a \emph{framework for synthesis} should follow a similar scheme.

\begin{mdframed}[innerleftmargin = 3pt, innerrightmargin = 3pt, skipbelow=-0.25em]
\begin{itemize}[labelsep=*, leftmargin=*]
  \item
    An instance of a synthesis problem $\textrm{Synthesize}(\mathcal{L}, \sem{\cdot}_{\mathcal{L}}, \varphi)$
    has three parameters:
    $\mathcal{L}$, a formal language;
    $\sem{\cdot}_{\mathcal{L}}$, a semantics to ascribe to $\mathcal{L}$; and
    $\varphi$, a behavioral specification for some desired member of $\mathcal{L}$.
    The behavioral specification $\varphi$ changes more frequently
    than $\mathcal{L}$ and $\sem{\cdot}_{\mathcal{L}}$.
  \item
    Let $F_{\textrm{syntax}}$ and $F_{\textrm{semantics}}$ be appropriate formalisms
    for specifying $\mathcal{L}$ and $\sem{\cdot}_{\mathcal{L}}$, respectively.
  \item
    Create a tool that implements the following specification:
    \begin{itemize}
      \item
        Input:
        an $F_{\textrm{syntax}}$ specification of a language's syntax, and
        an $F_{\textrm{semantics}}$ specification of the language's semantics.
      \item
        Output: a function $\textrm{Synth}_{\mathcal{L},\sem{\cdot}_{\mathcal{L}}}(\cdot)$
        such that $\textrm{Synth}_{\mathcal{L},\sem{\cdot}_{\mathcal{L}}}(\varphi)$ computes
        $\textrm{Synthesize}(\mathcal{L}, \sem{\cdot}_{\mathcal{L}}, \varphi)$
    \end{itemize}
\end{itemize}
\end{mdframed}

\noindent
As in \sygus, the formalism $F_{\textrm{syntax}}$ used in \semgus is
regular-tree grammars.
In \semgus, $F_{\textrm{semantics}}$ is Constrained Horn Clauses,
which are a class of logical formulas that are expressive enough to
define a recursive big-step semantics.
(In contrast, \sygus has no explicit formalism $F_{\textrm{semantics}}$;
instead, it relies on a shallow embedding into some decidable theory.)
\cbend

\Omit{\semgus formalizes the concept of semantics through 
Constrained Horn Clauses, which are a 
class of logical formulas that are expressive enough 
to define a recursive big-step semantics.}
The flexibility of Constrained Horn Clauses allows \semgus
to address synthesis problems for imperative programming languages.
In \S\ref{Se:MotivatingExample}, we show how the semantics
of assignments and while loops can be defined in \semgus.
The customizable aspect of the semantics also provides 
a natural way to define synthesis problems over an
alternative semantics (see \S\ref{Se:semgus}).
In essence, \semgus extends the ``logical framework'' 
of \sygus towards semantics, resulting in a framework 
that is capable of defining \sygus problems, as well as problems that currently require a solver-aided language.

\mypar{Solving \semgus Problems} Following the definition of the \semgus framework, this paper 
develops a method for solving general \semgus problems capable of producing 
\emph{two-sided answers} to a problem: either \emph{synthesizing a solution}, 
or proving that the problem is \emph{unrealizable}, i.e., has no solution.
Proving the unrealizability of synthesis problems has applications 
in synthesizing programs that are optimized with respect to some 
metric \cite{qsygus}, and can 
be employed in tandem with general synthesis algorithms as well.
However, existing program synthesizers are generally unable to 
prove unrealizability, and focus only on synthesizing terms.

Although \semgus can be used for much more than imperative program synthesis, 
solving \semgus problems over an imperative programming language 
illustrates many of the challenges in computing solutions to 
general \semgus problems:



\begin{description}
  \item[Reasoning while lacking a direct background theory.]
    Unlike \sygus, in which problems are defined over
    decidable theories, such as linear integer arithmetic,
    \semgus over an imperative programming language must 
    deal with factors such as state, and there is typically no
    decidable theory of the programming language involved.
  \item[Loops.]
    Loops provide a double challenge in the context of program
    synthesis:
    each loop could have (i) an infinite number of syntactic elaborations (of the condition and the loop-body),
    each of which may execute for (ii) an arbitrary number of iterations.
    Thus, a synthesis algorithm must reason about
    \emph{sets} of loop-body elaborations instead of \emph{individual}
    ones---otherwise, the search space becomes intractable.
    Existing constraint-based methods often deal with loops by setting an unrolling bound, which is a factor that 
    limits the kinds of synthesis problems they can define or solve: \semgus explicitly avoids this approach.
\end{description}


In \S\ref{Se:MotivatingExample} and \S\ref{Se:encoding}, we show that 
an entire \semgus problem---syntax, semantics, 
and behavioral specification---can be encoded using CHCs, effectively 
reducing program synthesis into a proof-search problem 
that can be solved with off-the-shelf CHC solvers, such as Z3~\cite{z3}.\footnote{
  In general, the problem of finding a solution to a set of CHCs is undecidable.
}
If a proof for the specification exists within the 
CHC-encoded syntax and semantic rules, the \semgus problem is 
realizable, and the proof identifies a specific term 
satisfying the specification.
If, on the other hand, the solver can prove that the specification is
unsatisfiable using the given rules, then the problem is unrealizable.
\semgus is semantics-guided not only in the sense that it accepts a semantics, 
but in this proof-search step as well: among the lemmas established during 
the proof search (by an external solver), some may involve the semantics 
supplied to \semgus by the user.

\mypar{Contributions} 
This paper makes the following contributions:
\begin{itemize}
  \item
    The \semgus framework, which allows the user to supply inference rules that
    specify the syntax and semantics of the target language.
    In particular, the \semgus framework can be used to specify synthesis problems over
    an imperative programming language (\S\ref{Se:semgus}).
  \item
    A constraint-based approach for solving \semgus problems using CHCs 
    (\S\ref{Se:MotivatingExample} and \S\ref{Se:encoding}), capable of producing 
    both a synthesized program for realizable problems, and a proof of unrealizability 
    for unrealizable ones.
  \item 
    Multiple instantiations of the framework---with different kinds of semantics---that express 
    variant \semgus problems whose solutions can sometimes be obtained more efficiently (\S\ref{Se:Semantics}).
  \item
    An implementation of a \semgus solver using Z3~\cite{z3,spacer},
    called \ourtool.
    We instantiate \ourtool to come with a variety of semantics out-of-the-box, 
    allowing users to easily define and solve \semgus problems.
    Moreover, \ourtool
    is the first tool capable of both (i) solving synthesis problems,
    and (ii) proving unrealizability for imperative-language
    problems involving a search space with an infinite number
    of programs.
\end{itemize}
\S\ref{Se:problem-def} provides background material.
\S\ref{Se:related} discusses related work.
\S\ref{Se:conclusion} concludes.
Proofs for the various theorems throughout the paper can be found in Appendix~\ref{Appendix:Proofs}.



\section{Motivating Example}
\label{Se:MotivatingExample}


Consider the problem of synthesizing an imperative
 program that stores
the bitwise-xor of two variables $x$ and $y$ in the variable $x$, 
using only 
bitwise-and and bitwise-or operations and no auxiliary variables.
We show how one can define this problem in \semgus and prove it
unrealizable.

\subsection{Defining a \semgus Problem}
\label{Se:semgus-def}

The first contribution of this paper is the \semgus framework (\S\ref{Se:semgus}).
A \semgus problem
is defined using three components: (i) a search space 
given by a regular tree grammar $G$, (ii) a semantics for the grammar 
$G$, and (iii) a specification of the desired behavior of the program.

\mypar{Supplying \semgus with a grammar}
In this example,
the grammar $G_{ex}$ in Figure~\ref{fig:gex} describes a language of
single-loop programs
that can contain an arbitrary number of assignments to
$x$ and $y$, but involve only bitwise-and and bitwise-or operations.
In the figure, the numbers in the black circles are used as 
unique identifiers for each production.
Note that \sygus cannot describe the language $L(G_{ex})$  due
to the presence of assignments and loops.

\begin{figure}[t!]
  $$
\begin{array}{rlll}
  \mathit{Start} & ::= &  \Ewhile{B}{S}  \ball{1}\\
  					  B & ::= &  E < E  \ball{2}\\
  					  S & ::= & \Eseq{S}{S}  \ball{3} \mid   \Eassign{x}{E}  \ball{4} \mid   \Eassign{y}{E}   \ball{5}\\ 
  					  E & ::= & x  \ball{6} \mid  y  \ball{7} \mid E~\&~E  \ball{8} \mid E~|~E  \ball{9}
\end{array}
  $$
\caption{Example grammar $G_{ex}$.}\label{fig:gex}
\end{figure}

\mypar{Supplying \semgus with a semantics}
The next component of a \semgus problem is a semantics for terms in the language $L(G_{ex})$.
There are many possible ways to define the formal semantics of an
imperative language.
For example, if we let
 $\Gamma$, $\Gamma_1$, and $\Gamma_2$ denote valuations of the variables $x$ and $y$,
Equation~\eqref{eq:while-true} defines a semantics that a user might give 
for the term ``$\Ewhile{b}{s}$''.
\begin{equation}
\label{eq:while-true}
  \infer[sem^{\mathit{WTrue}}]
{\sem{\Ewhile{b}{s}}(\Gamma,\Gamma_2)}
{
  \sem{b}(\Gamma, \Etrue) 
  \quad
  \sem{s}(\Gamma,\Gamma_1)
  \quad 
  \sem{\Ewhile{b}{s}}(\Gamma_1,\Gamma_2)
}
\end{equation}

Equation~\eqref{eq:while-true} is a common way to define program semantics, but it contains 
 some ambiguities, such as the method of defining the semantic function $\sem{\cdot}$.
\semgus takes an extra level of formalization and requires that semantic rules such as 
$sem^{\mathit{WTrue}}$ are expressed using \emph{logical relations} and \emph{Constrained Horn Clauses} (CHCs), 
which are implications 
that are defined over logical relations and a single logical constraint.
As an example, we show how $sem^{\mathit{WTrue}}$ can be expressed as a CHC in \semgus.

In \semgus, semantics can be specified in a compositional fashion by associating each \emph{production} 
in the grammar with one or more semantic rules,\footnote{
  The ability to define multiple semantic rules for a production is useful when defining semantics for productions
  such as $\mathit{Start} \rightarrow \Ewhile{b}{s}$, which is commonly equipped with two rules that describe looping and loop termination.
}
with the additional constraint that each rule must be expressible as a CHC.
To express the semantics of terms, which are derived from each nonterminal in a
production, we assume that each nonterminal $N$ has a corresponding logical relation 
$\mathsf{sem}_N$, which represents the behavior of the semantic function 
$\sem{\cdot}$ in Equation~\eqref{eq:while-true}.
We refer to these relations as \emph{semantic relations}.
For example, the expression $\BFlreduceTree{\Gamma, b}{v_b}$ corresponds to the 
premise $\sem{b}(\Gamma, v_b)$: the semantic relation 
$\mathsf{sem}_B$ tells us that executing the term $b \in L(B)$ on incoming state $\Gamma$
results in a value $v_b$.

\begin{equation}
  \label{eq:semantics_ex}
\infer[\mathit{sem}^{\rWT}_{\mathit{Start} \rightarrow  \Ewhile{B}{S}\sball{1}}]
  {\StartlreduceTree{\Gamma, \Ewhile{b}{s}}{\Gamma_{2}}}
{
\begin{array}{l}
  \BFlreduceTree{\Gamma, b}{\Etrue} \ \ \ 
  \SFlreduceTree{\Gamma, s}{\Gamma_1} \ \ \ 
  \StartlreduceTree{\Gamma_1, \Ewhile{b}{s}}{\Gamma_{2}}
\end{array}
}
\end{equation}

Equation~\eqref{eq:semantics_ex} uses semantic relations to express
the same semantics as Equation~\eqref{eq:while-true}, and fits our
criterion of using CHCs as semantic rules:
the relations $\mathsf{sem}_B, \mathsf{sem}_S, $ and $\mathsf{sem}_{Start}$ represent the 
semantics of terms, while the whole of Equation~\eqref{eq:semantics_ex} can be read as a CHC.

\semgus assumes  the supplied semantics are of the form in Equation~\eqref{eq:semantics_ex}:
that is, the semantic relations model the semantics of terms,
and each semantic inference rule is a CHC.\footnote{
  For clarity, we sometimes use the format from Equation~\eqref{eq:while-true} when introducing 
   semantics for \semgus to work with (\S\ref{Se:Semantics}).
}
This assumption is not restrictive, nor does it impose a complex format---CHCs are expressive enough to 
model a recursively defined big-step semantics.
Rather than restricting the kinds of problems \semgus supports, these restrictions mainly 
exist to formalize the meaning of the word ``semantics''.

\mypar{Supplying \semgus with a behavioral specification}
The behavioral specification of \semgus states what property the target program should 
satisfy.
One can provide a logical formula that relates the input and output
valuations of the program variables, or alternatively, provide
a set of examples, which is the kind of specification our algorithm for 
solving \semgus problems relies upon.\footnote{
  Given a logical specification, one can always generate a set of examples 
  and add more examples as needed through a technique known as 
  \emph{counterexample-guided inductive synthesis} (CEGIS), which is 
  applied in many synthesizers.
}

For our example, suppose that the specification is given as the 
set of input valuations $[(6, 9)$, $(44, 247)$, $(14, 15)]$
(each representing the values of $x$ and $y$, respectively), which produce the output values 
$[15, 219, 1]$ (each representing the final value of $x$).
Call this example set $\examples_{ex}$.
In \S\ref{Se:semgus-solve}, we show how our  algorithm (implemented in \ourtool) can 
synthesize a valid solution on a \emph{subset} of $\examples_{ex}$, namely,
$[(6, 9)]$ with output $[15]$, where bitwise-xor is equivalent to bitwise-or, i.e., this sub-problem is realizable.
We then describe how our algorithm proves that 
\emph{no} program in the language of $L(G_{ex})$
 can compute the bitwise-xor
for all the examples in $\examples_{ex}$, 
i.e., that the problem is unrealizable.



\subsection{Solving \semgus Problems}
\label{Se:semgus-solve}

The second contribution of this paper is a procedure for solving
\semgus problems (\S\ref{Se:encoding}).
To solve a \semgus problem, this paper utilizes two key ideas:
(i) both the syntax and the semantics of a synthesis problem can be described 
using Constrained Horn Clauses, and, 
(ii) one can phrase the synthesis problem as a proof search over CHCs.

\mypar{Syntax and Semantic Rules} 
Describing a grammar using CHCs is a straightforward process: 
taking the production 
$\mathit{Start} \rightarrow \Ewhile{B}{S} \ball{1}$ as an example, the production 
states that one can obtain a valid term for the nonterminal $\mathit{Start}$ using 
valid terms for nonterminals $B$ and $S$.
Equation~\eqref{eq:syntax_ex} encodes this property as a CHC.
\begin{equation}
\label{eq:syntax_ex}
  \infer[\textit{syntax}_{\mathit{Start} \rightarrow \Ewhile{B}{S}\sball{1}}]
  {\StartrreduceTree{\Ewhile{b}{s}}}
  {
    \begin{array}{l}
    \BFrreduceTree{b} \quad
    \SFrreduceTree{s}
    \end{array}
  }
\end{equation}
The logical relations $\mathsf{syn}_B$, 
$\mathsf{syn}_S$, and $\mathsf{syn}_{Start}$ in Equation~\eqref{eq:syntax_ex} 
model whether the supplied arguments are 
valid terms that may be derived from the corresponding nonterminals $\mathit{B}$, $\mathit{S}$, and $\mathit{Start}$.
We refer to relations such as $\mathsf{syn}_{S}$ as \emph{syntax relations}, and 
rules such as Equation~\eqref{eq:syntax_ex} as \emph{syntax rules}.

\S\ref{Se:semgus-def} illustrated how the programming-language semantics can be expressed using CHCs;
in tandem with the syntax rules, they represent the semantics of all possible programs 
in the language.

\mypar{Specification Query}
The final step to solving a \semgus problem is to create a query 
that encodes the behavioral specification, asking whether any of the programs 
generated by the grammar is consistent with the specification 
on the set of input examples $\examples$.
This question is posed via the $\mathsf{Query}$ rule below, which checks for the existence
of a term $t$ that satisfies the syntax rules and the semantic rules, each instantiated with 
input $e_i \in \examples$
and corresponding output value $o_i$.\footnote{
  In practice, we encode terms into an alternative representation
  because SMT solvers have difficulties when terms are expressed directly.
  This encoding is presented in $\S\ref{Se:encoding}$.
}

\begin{equation}
  \label{query-motive}
\infer[\mathsf{Query}]{\mathit{Realizable}}
{
  \StartrreduceTree{t} \quad
  \bigwedge_{e_i \in \examples} \StartlreduceTree{e_i, t}{o_i}
}
\end{equation}

Generally, one could choose to use symbolic variables for $o_i$ instead of concrete output examples, 
by adding an additional premise 
$\bigwedge_{e_i \in \examples} \psi(e_i, o_i)$ to ensure that the input-output pair 
$e_i, o_i$ meets the specification $\psi$.
In this section, we consider concrete output examples for ease of presentation.

Expressing the entire \semgus problem as a set of inference rules and a query
effectively reduces solving the \semgus problem to a proof search to establish that 
$Realizable$ holds using the 
given inference rules.
If one can prove that the premises of Equation~\eqref{query-motive} hold,
then the \semgus problem is realizable, and the term $t$ is a concrete answer to the problem.
If there exists no proof for $Realizable$ using the inference rules, 
then the \semgus problem is unrealizable.

\mypar{Synthesizing Programs}
\cbstart
\jinwoo{Shepherd Item 1: Add a fully worked out example}
To see how a valid program is synthesized based on our construction, consider the problem of
synthesizing a program that computes the bitwise-xor of $x$ and $y$, specified using the 
singleton example set $[(x, y)] = [(6, 9)]$ and output $[15]$.
In this case, the CHC solver is responsible for finding a term $t$ 
that satisfies the conjunction of the relations (and, as stated above, also corresponds to proving): 
\[
  \begin{array}{l@{\hspace{5ex}}l}
    \StartrreduceTree{t} & \Startlreduce{(6,9), t}{(15, y')} \\
	\end{array}
\]

\begin{figure}
{ \tiny
$$
  \infer[]{\StartrreduceTree{t}}{
    \infer[]{\BFrreduceTree{x < y}}{
      \infer[]{\EFrreduceTree{x}}{} \:
      \infer[]{\EFrreduceTree{y}}{}
    } \:
    \infer[]{\SFrreduceTree{\Eassign{x}{x~|~y}}}{
      \infer[]{\EFrreduceTree{x~|~y}}{
        \infer[]{\EFrreduceTree{x}}{}\:
        \infer[]{\EFrreduceTree{y}}{}
      }
    }
  } \:
  \infer{\Startlreduce{(6, 9), t}{(15, y')}}{
    \infer[]{\BFlreduce{(6, 9), x < y}{\Etrue}}{(6 < 9) = \Etrue} \:
    \infer[]{\SFlreduce{(6, 9), \Eassign{x}{x~|~y}}{(15, 9)}}{
      \infer[]{\EFlreduce{(6, 9), x~|~y}{15}}{
        6~|~9 = 15
      }
    } \:
    \infer[]{\Startlreduce{(15, 9), t}{(15, 9)}}{
      \infer[]{\BFlreduce{(15, 9), x < y}{\Efalse}}{(15 < 9) = \Efalse}
    }
  } 
  \vspace{-1.7mm}
  $$
  \vspace{-2.2mm}
  $\hspace{-37mm} \vcenter{\rule{0.65\textwidth}{0.4pt}} \hspace{-23mm} \mathsf{Query}$
  $$
  \hspace{-25mm} \mathit{Realizable}
$$
  \vspace{-6mm}
}
  \caption{The full proof tree for synthesizing the bitwise-xor of $x$ and $y$ from
  the input example $[(x, y)] = [(6, 9)]$ with output $[15]$, using the grammar $G_{ex}$.
  The term $t = (\Ewhile{x < y}{\Eassign{x}{x~|~y}})$ satisfies the one example provided: 
  the loop iterates once to set $x$ to $15$.}
  \label{fig:full_ex}
\end{figure}

\noindent
The $\SemStart$ literal states that the final value of $x$ should be $15$, and, via free variable $y'$, that we do not care about the final
value of $y$.
(The reason why the final value of $y$ is also present in the relation
is because nonterminal $\mathit{Start}$ represents a statement,
and thus the relation $\SemStart$ must track changes to both
$x$ and $y$.)
For this input/output pair, bitwise-xor is indistinguishable from bitwise-or ($~|~$), making
the problem realizable:
the term $t = (\Ewhile{x < y}{\Eassign{x}{x~|~y}})$ is a solution.
Figure~\ref{fig:full_ex} shows how such a solution corresponds to a proof in our system of 
CHCs, where the syntax premise $\StartrreduceTree{t}$ ensures that $t$ is indeed a valid term, 
and the semantic premise $\Startlreduce{(6,9), t}{(15, y')}$ ensures that the semantics of $t$ 
matches the specification.
Our tool \ourtool (which is based on Z3~\cite{z3} and its CHC solver Spacer~\cite{spacer}) 
succeeds in deriving the proof tree in Figure~\ref{fig:full_ex}, from which the term $t$ is then extracted.

\cbend

\mypar{Proving Unrealizability}
To see how a \semgus problem is proved unrealizable, recall our full example set $\examples_{ex}$, for which
the solver must find some term $t$ that satisfies the relations:
\[
  \begin{array}{l@{\hspace{5ex}}l}
    \StartrreduceTree{t} & \StartlreduceTree{(6,9), t}{(15, y'_1)} \\
    \StartlreduceTree{(44,247), t}{(219, y'_2)} & \StartlreduceTree{(14,15), t}{(1, y'_3)}
  \end{array}
\]
Put another way, the solver must establish that there
exists \emph{no} term $t$ that satisfies all four relations at once---i.e.,
that $Realizable$ is $\emph{unsatisfiable}$---to prove the problem unrealizable.

Note that our algorithm does not provide
additional machinery to reason about loops.
Instead,  we rely on the
CHC solver to discover lemmas about \emph{sets of loops}---as
opposed to single loops---to prune the search space.
\cbstart
\jinwoo{Extension of concrete example above: show concrete lemma}
When attempting to find a proof tree for $\mathit{Realizable}$ consistent with the examples in $\examples_{ex}$, 
the CHC solver Spacer eventually proves the following invariant for the third example, namely $(14, 15) \rightarrow 1$:
``For all values of $x$ that are reachable from the nonterminal $\mathit{Start}$, $x~\&~4 = 4$ always holds''.
The lemma $x~\&~4 = 4$ implies that in the theory of fixed-length bitvectors, the 
\emph{third bit of} $x$ must \emph{always} be $\Etrue$ when the loop terminates.
This condition conflicts with the output $1$ (in which the third bit is $\Efalse$),
which shows that the third example can never be satisfied---which, in turn, implies that the
synthesis problem is unrealizable!
Note that this lemma is an invariant of the \emph{nonterminal}
$\mathit{Start}$---i.e., an invariant of \emph{all} programs derivable from
$\mathit{Start}$---not just some specific program derivable from
$\mathit{Start}$.
\cbend

One might be tempted to give an operational reading of the $\mathsf{Query}$ rule as
following the paradigm of \emph{generate and test}:
$\StartrreduceTree{t}$ generates $t$, which then must pass the tests
$\StartlreduceTree{I_1, t}{ O_1}\, \ldots\, \StartlreduceTree{ I_n, t}{ O_n}$.
However, the ability of Spacer to prove lemmas of the sort discussed above means that
\ourtool is not merely enumerating and testing individual programs.
On the contrary, the technique for solving \semgus problems
infers lemmas about the behavior of \emph{multiple} programs in the language of the grammar,
and uses them to prune the search space!


\subsection{Instantiating \semgus with Other Semantics}
\label{Se:semgus-other}

The procedure described in the previous section gives a general way to solve \semgus problems, 
but also suffers from several limitations.
For example, one might have to prove a large number of $\mathsf{sem}$ relations from the premise 
of the $\mathsf{Query}$ rule if there are a large number of input-output examples; or, because 
solving CHCs is still difficult in general, the problem may simply be too difficult to solve.
As a third contribution, we show how, thanks to its generality,
\semgus can be supplied with alternative semantics to 
address some of these challenges (\S\ref{Se:Semantics}). 
As an example, here we show how to supply \semgus with an
\emph{abstract semantics} to prove unrealizability more efficiently.

Consider again the problem of proving that synthesizing a 
bitwise-xor program from the grammar $G_{ex}$ is unrealizable.
As described in \S\ref{Se:semgus-solve}, the lemma used to prove this fact states that 
the third bit of $x$ under the example $(14, 15) \rightarrow 1$ is always set to $\Etrue$, 
conflicting with the output $1$.
While we proved this problem unrealizable using a precise semantics, 
it is also possible to prove unrealizability using an abstract domain.
For example, 
consider the abstract domain $\mathbb{B}_3$, which only tracks the value of the third bit of every variable,
using the values $\Etrue, \Efalse, $ and $\top$ (top), 
where $\top$ represents the scenario in which the third bit may be either 
$\Etrue$ or $\Efalse$: i.e., the semantics may be imprecise.
Then, one could supply an abstract semantics for a term $e_1~\&~e_2$ (the bitwise-and of $e_1$ and $e_2$), created from 
the production $E \rightarrow E~\&~E$, as:

\begin{equation}
\label{eq:abstract-and}
  \infer[\mathsf{And^{\#}}]
  {\sem{e_1~\&~e_2}^\# (\Gamma^\#, v^\#)}
{
  \sem{e_1}^\# (\Gamma^\#, v_1^\#)
  \quad
  \sem{e_2}^\# (\Gamma^\#, v_2^\#)
  \quad
  v^\# = (\Eifthenelse{(v_1^\# = \top \vee v_2^\# = \top)}{\top}{v_1^\#~\&~v_2^\#})
}
\end{equation}

The final premise in Equation~\eqref{eq:abstract-and} represents the abstract 
transformer of bitwise-and $\mathbb{B}_3$, which 
sends the computation to $\top$ if any of $v_1^\#$ or $v_2^\#$, 
the abstract values for $v_1$ and $v_2$, are $\top$, or 
computes the exact value otherwise.
$\top$ can be generated in $\mathbb{B}_3$ by operators such as $+$, 
which always loses precision because it does not track carry bit values
from the second position.


From a \semgus point of view, an abstract semantics  is merely a different semantics, 
which allows \semgus problems with abstract semantics such as $\mathbb{B}_3$ to be solved using the same 
algorithm described in \S\ref{Se:semgus-solve}.
Although $\mathbb{B}_3$ is more lightweight compared to the 
precise semantics discussed in \S\ref{Se:semgus-def}, it is sufficient to 
prove the unrealizability of synthesizing bitwise-xor from $G_{ex}$---therefore 
resulting in a more efficient solving procedure.

In \S\ref{Se:Semantics}, we show how other semantics, such as an
underapproximating one, can  be supplied to the \semgus framework, 
each with their advantages.
These semantics illustrate one of the  benefits of allowing a 
user to supply their own semantics in \semgus---in addition to a wider range 
of definable problems, one can also describe specific strategies to 
optimize the synthesis problem at hand!



\section{Preliminaries}
\label{Se:problem-def}

In this section, we provide some background information on concepts that 
we build upon for the rest of the paper.
\S\ref{Se:horn-clause} provides background on Horn Clauses, 
which are used in \S\ref{Se:encoding} to define our procedure for solving \semgus problems.
\S\ref{Se:trees-and-semantics} is about trees, regular tree grammars, 
and program semantics, which are required for our definition of the 
\semgus problem in \S\ref{Se:semgus}.

\subsection{Constrained Horn Clauses}
\label{Se:horn-clause}

\emph{Constrained Horn Clauses} (CHCs) are a class of logical rules that we use to 
formalize the concept of semantics, as well as use in our algorithm for solving 
\semgus problems.

\begin{definition}[Constrained Horn Clauses.]
  \label{Def:horn-clause}
  A \emph{Constrained Horn Clause} is a first-order formula of the form 
  $\forall \vec{x}, \vec{x_1},  \ldots, \vec{x_n}. 
  (\phi \wedge R_1(\vec{x_1}) \wedge \cdots \wedge R_n(\vec{x_n}) \implies H(\vec{x}))$, 
  where $\phi$ is a constraint over some background theory that may contain variables from 
  $\vec{x}, \vec{x_1}, \ldots, \vec{x_n}$, and 
  $R_1, \ldots, R_n$ and $H$ are uninterpreted relations.
\end{definition}

\cbstart
\jinwoo{Review B: First-order terms in CHC arguments}
CHCs often allow first-order terms directly in their arguments, which can be viewed as a form 
of syntactic sugar with respect to Definition~\ref{Def:horn-clause}: for a first-order formula $f(x)$,
one can rewrite $R(f(x))$ to $R(y)$ and add the constraint $f(x) = y$ to $\phi$.
For the remainder of the paper, we assume that CHCs allow first-order terms as their arguments.
\cbend

\begin{example}
  \label{ex:horn-clause}
  Equations~\eqref{eq:syntax-horn-example} and~\eqref{eq:semantic-horn-example} give an example of how the syntax and semantic rules from 
  \S\ref{Se:MotivatingExample} can be interpreted as CHCs.
  \begin{equation}
    \label{eq:syntax-horn-example}
    \forall b, s. \:
    \BFrreduceTree{b} \wedge
    \SFrreduceTree{s} \implies 
    \StartrreduceTree{\Ewhile{b}{s}}
    \vspace{2mm}
  \end{equation}
  \begin{equation}
    \label{eq:semantic-horn-example}
    \begin{split}
      \forall \Gamma, \Gamma_1, \Gamma_2, b, s. \: & 
      (v_b=true) \wedge
      \BFlreduceTree{\Gamma, b}{v_b} \\
      & \wedge \: \SFlreduceTree{\Gamma, s}{\Gamma_1} \wedge 
      \StartlreduceTree{\Gamma_1, \Ewhile{b}{s}}{\Gamma_2} \\
       & \implies \StartlreduceTree{\Gamma, \Ewhile{b}{s}}{ \Gamma_2}
    \end{split}
  \end{equation}
  Syntax and semantic relations 
  such as $\mathsf{syn}_B$ or $\mathsf{sem}_{Start}$ are 
  expressed as uninterpreted relations, while atomic semantic operations such as addition 
  are represented using the constraint $\phi$.
\end{example}

CHCs occur frequently in program verification, and many efficient algorithms 
for solving CHCs have been developed~\cite{vampire, chcinterpolation, spacer}.
In terms of CHC solving, the proof search described in \S\ref{Se:MotivatingExample} 
is augmented with an extra CHC $\mathit{Realizable} \implies \Efalse$, which asserts $\lnot \mathit{Realizable}$.
If there exists an interpretation of the syntax and semantic relations that can derive $\mathit{Realizable}$, 
then the system of CHCs is not valid---and, as shown in \S\ref{Se:MotivatingExample}, 
one can extract a program from the proof tree for $\mathit{Realizable}$.
If $\mathit{Realizable}$ cannot be derived from any interpretation of the syntax and semantic relations, 
the system of CHCs is valid, which corresponds to unrealizability.
In this paper, we use the (un)satisfiability of $\mathit{Realizable}$ itself, instead of the validity of the augmented CHC problem, 
to illustrate our algorithm.

\subsection{Trees, Tree Grammars, and Semantics}
\label{Se:trees-and-semantics}

A \emph{ranked alphabet} is a tuple $(\Sigma,\rk_\Sigma)$, where
$\Sigma$ is a finite set of symbols, and $\rk_\Sigma:\Sigma\to\nat$
associates a rank to each symbol.
For every $m\ge 0$, the set of all symbols in $\Sigma$ with rank $m$
is denoted by $\Sigma^{(m)}$.
In our examples, a ranked alphabet is specified by showing the set
$\Sigma$ and attaching the respective rank to every symbol as a
superscript---e.g., $\Sigma=\{+^{(2)},c^{(0)}\}$.
(For brevity, the superscript is often omitted.)
We use $T_\Sigma$ to denote the set of all (ranked) trees over
$\Sigma$---i.e., $T_\Sigma$ is the smallest set such that
\rone $\Sigma^{(0)} \subseteq T_\Sigma$,
\rtwo if $\sigma^{(k)} \in \Sigma^{(k)}$ and $t_1,\ldots,t_k\in T_\Sigma$, then
$\sigma^{(k)}(t_1,\cdots,t_k)\in T_\Sigma$.
In what follows, we assume a fixed ranked alphabet $(\Sigma,\rk_\Sigma)$.

In this paper, we focus on \emph{typed} regular tree grammars, in which
each nonterminal and each symbol is associated with a type.
There is a finite set of types $\{ \type_1, \ldots, \type_k \}$.
Associated with each symbol $\sigma^{(i)} \in \Sigma^{(i)}$, there is a type assignment
$a_{\sigma^{(i)}} = (\type_0, \type_1, \ldots, \type_i)$, where $\type_0$ is called the \emph{left-hand-side type}
and $\type_1, \ldots, \type_i$ are called the \emph{right-hand-side types}.
Tree grammars are similar to word grammars, but generate trees over a ranked alphabet
instead of words. 


\begin{definition}[Regular tree grammar]
A \emph{typed regular tree grammar} (RTG) is a tuple $G = (N,\Sigma,S,a,\delta)$, 
where $N$ is a finite set of non-terminal symbols of arity 0; 
$\Sigma$ is a ranked alphabet; $S \in N$ is an initial nonterminal; 
$a$ is a type assignment that gives types for members of $\Sigma \cup N$; 
and $\delta$ is a finite set of productions of the form 
$A_0 \rightarrow \sigma^{(i)}(A_1,...,A_i)$, where for $1 \leq j \leq i$, each $A_j \in N$ is a nonterminal 
such that if $a_{\sigma^{(i)}} = (\type_0,\type_1,...,\type_i)$ then $a_{A_j} = \type_j$.
\end{definition}
Given a tree $t\in T_{\Sigma\cup N}$,
applying a production $r = A\to\beta$ to $t$ produces the tree $t'$ resulting from replacing
the leftmost occurrence of $A$ in $t$ with the right-hand side $\beta$.
A tree $t\in T_\Sigma$ is generated by the grammar $G$---denoted by $t\in
L(G)$---iff it can be obtained by applying a sequence of
productions $r_1\cdots r_n$ to the tree whose root is the initial
non-terminal $S$.
\Omit{Similarly, for $A \in N$, $L(A)$ denotes the tree language generated from nonterminal $A$
(i.e., $L(A) = L(G')$, where $G' = (N,\Sigma,A,a,\delta)$).}

Figure~\ref{fig:gex} from \S\ref{Se:MotivatingExample} shows an example of a typed regular tree grammar.
For readability, the grammar does not contain explicit symbols---e.g.,
the production $Start \rightarrow \Ewhile{B}{S}$ should be more correctly stated as
a production $Start \rightarrow \mathsf{while}(B,S)$, where $\mathsf{while}$ is a binary symbol. 
We will use the former notation for readability, and assume that all expressions are well-typed.

We note that terms can be represented using trees of productions, which 
makes it easier to distinguish terms created by different productions 
with identical operators: we use this representation in our tool \ourtool.

\begin{example}
  \label{ex:term-tree}
  Recall the grammar $G_{ex}$ from Figure~\ref{fig:gex}, where each production is labeled with a unique identifier $\sballlrns{n}$.
  The term ``$\Ewhile{x < x}{\Eassign{x}{y}}$'' can be represented using the tree
  $\bTree{\sballlrns{1}}{\bTree{\sballlrns{2}}{\tNode{\sballlrns{6}}}{\tNode{\sballlrns{6}}}}{\uTree{\sballlrns{4}}{\tNode{\sballlrns{7}}}}$.
  The first child tree $\bTree{\sballlrns{2}}{\tNode{\sballlrns{6}}}{\tNode{\sballlrns{6}}}$ represents the condition ``$x < x$'', while
  the second child tree $\uTree{\sballlrns{4}}{\tNode{\sballlrns{7}}}$ represents the assignment ``$\Eassign{x}{y}$''.
\end{example}


When defining a \semgus problem, one has to provide a semantics for the productions
in the RTG.
The semantic definitions are allowed to use terms from a theory $\mathcal{T}$ (e.g., linear integer arithmetic).

\cbstart
\jinwoo{Review B: Replace $v$ with $\Upsilon$ (because $v$ is evocative of a value) / added explanation on signatures of CHCs}

\begin{definition}[Production-based semantics]
\label{Def:prod-semantics}
  Given an RTG $(N, \Sigma, S, a, \delta)$ and a theory $\mathcal{T}$,
  a \textit{semantics for the grammar}
  is a function $\sem{\cdot}$ that maps every production
  $A_0 \rightarrow \sigma^{(i)}(A_1,...,A_i)$ of type
  $a_{\sigma^{(i)}} = (\type_0,\type_1,...,\type_i)$
  to a set of Constrained Horn Clauses of the form 
  $\phi \wedge \mathsf{sem}_{A_1}(\Gamma_1, t_{A_1}, \Upsilon_1) \wedge \cdots 
  \mathsf{sem}_{A_i}(\Gamma_i, t_{A_i}, \Upsilon_i) \implies \mathsf{sem}_{A_0}(\Gamma_0, t_{A_0}, \Upsilon_0)$,
  where $\mathsf{sem}_{A_0}, \mathsf{sem}_{A_1}, \cdots \mathsf{sem}_{A_i}$ are uninterpreted 
  relations, $\Gamma_0, \Gamma_1, \cdots \Gamma_i$ are variables that represent state, 
  $t_{A_k}$ is a variable that represents a term $t \in L(A_k)$, 
  $\Upsilon_0, \Upsilon_1, \cdots \Upsilon_i$ are variables of type $\type_0, \type_1, \cdots \type_i$,
  and $\phi$ is a constraint within the theory $\mathcal{T}$.
\end{definition}

The function $\sem{\cdot}$ can be lifted to trees as follows:
for every subtree $t'$ of $t$,
if $t'=\sigma^{(i)}(t_1,...,t_i)$, then $\sem{t'}=\sem{\sigma^{(i)}}(\sem{t_1},\ldots,\sem{t_i})$.
In practice, the type signatures for $\Gamma$ and $\Upsilon$ are
nearly unrestricted---one has great flexibility in defining for
each production
(i) the kind of program state that is tracked in $\Gamma$, as well as
(ii) the kind of value returned in $\Upsilon$.
This approach allows one to use arbitrary CHCs to define the
semantics, as long as it contains the term $t$ as an argument.

\begin{example}
  The semantics for the statement nonterminal $S$, from Figure~\ref{fig:gex} and Equation~\eqref{eq:semantics_ex}, 
  uses the product type $\mathsf{Int} \times \mathsf{Int}$ for both $\Gamma$ and $\Upsilon$.
\end{example}
\cbend

\Omit{The semantic function $\sem{\cdot}$ can be lifted to partial derivation trees in
$T_{\Sigma\cup N}$ as follows:
the semantics of a nonterminal $A$ is just the set of meanings of terms in $L(A)$---i.e.,
$\sem{A} = \{ \sem{s} \mid s \in L(A) \}$.
Then, for every subtree $t'$ of $t \in T_{\Sigma\cup N}$,
if $t'=\sigma^{(i)}(t_1,...,t_i)$, then
$\sem{t'} = \{ \sem{\sigma^{(i)}}(u_1,\ldots,u_i) \mid u_1 \in \sem{t_1},\ldots, u_i \in \sem{t_i} \}$.}

As is common in many semantic definitions, Def.~\ref{Def:prod-semantics}
defines the semantics of terms in the grammar inductively.
This ability to \textit{equip the grammar with customized semantics} is the
defining characteristic that distinguishes \semgus from \sygus. 
In \sygus, the underlying theory---e.g., LIA---is what corresponds to the specified semantics.
In \semgus, the semantics can be any Constrained Horn Clause defined over 
the relations $\mathsf{sem}_{A_0}, \mathsf{sem}_{A_1}, \cdots \mathsf{sem}_{A_i}$.


\begin{example}
  The big-step semantics of simple imperative languages can be expressed using rules like the
  one illustrated in Equation~\eqref{eq:semantics_ex}, which inductively defines the semantic of 
  the production $Start \rightarrow \Ewhile{B}{S}$ through the semantic relations for nonterminals 
  $B$ and $S$.
\end{example}

\section{Semantics-Guided Synthesis and its Properties}
\label{Se:semgus}

We now provide a formal definition of the Semantics-Guided Synthesis problem:



\begin{definition}[\semgus]
   \label{def:semgus}
   A \textit{\semgus problem} over a theory $\mathcal{T}$
   is a tuple $sem = (G, \forall x. \psi(x, f(x)))$,
   where $G$ is a regular tree grammar with a production-based semantics $\sem{\cdot}$,
   and $\forall x.\psi(x, \sem{f}(x))$ is a Boolean formula over the theory $\mathcal{T}$
   that specifies the desired behavior of $f$, where $f$ is a free second-order variable.
   A \textbf{solution} to the  \semgus problem $sem$ is a term $s\in L(G)$ such that
   $\forall x.\psi(x, \sem{s}(x))$ holds.
   We say that $sem$ is \textbf{realizable} if a solution exists and \textbf{unrealizable} otherwise.
\end{definition}

\begin{example}
  \label{ex:semgus-xor}
  The problem of synthesizing a program for bitwise-xor 
  described in \S\ref{Se:MotivatingExample} can be written as a \semgus problem 
  $sem = (G_{ex}, \forall x, y. f(x, y) = x \oplus y)$ (with $\oplus$ denoting bitwise-xor), 
  where $G_{ex}$ is equipped with a semantics that contains the rule given in Equation~\eqref{eq:semantics_ex}.
\end{example}

Example~\ref{ex:semgus-xor} gives an example of a \semgus problem where the grammar 
is equipped with a semantics that one would normally expect for imperative programs.
Definition~\ref{def:semgus}, which defines \semgus problems, shows that \semgus can be instantiated 
with different kinds of semantics, as long as the semantics satisfies the definition of a
production-based semantics (Definition~\ref{Def:prod-semantics}).
This feature allows \semgus problems to be instantiated with a semantics that is \textit{approximate} with
respect to some original semantics.
An approximate semantics can be used to efficiently compute 
one-sided answers to the original problem---either synthesis or unrealizability---depending on the relation between the
approximating and the original semantics.

\subsection{Unrealizability of \semgus Problems with Overapproximating Abstract Semantics}
\label{Se:semgus-abstract}

In this section, we see how an \emph{overapproximating semantics}  can be used to prove 
\emph{unrealizability}.
An overapproximating semantics overapproximates the set of reachable states with 
respect to an original program semantics; in essence, they are an \emph{abstract semantics}~\cite{absint}, 
and we use the latter term for the rest of the paper.
More specifically, we show 
that if a \semgus problem $sem = (G, \psi(x, f(x)))$ is unrealizable 
when $G$ is equipped with an abstract semantics, then $sem$ 
is unrealizable when equipped with the original semantics as well.

\begin{definition}
  For a grammar $G$ equipped with a semantics $\sem{\cdot}$, we say $\sem{\cdot}^{\#}$ is an
  \emph{abstract semantics} for $G$ with respect to $\sem{\cdot}$ if there exists
  an abstraction function $\alpha$ and a concretization function $\gamma$, such that
  for all $t \in L(G)$, if
  $\sem{t}(\Gamma, v)$ holds, then
  $\sem{t}^{\#}(\alpha(\Gamma), \alpha(v))$ holds, and $\Gamma \in \gamma(\alpha(\Gamma))$, $v \in \gamma(\alpha(v))$,
  i.e., $\alpha$ and $\gamma$ form a Galois connection.
\end{definition}

In $\semgus$,  an abstract semantics $\sem{\cdot}^{\#}$ overapproximates
the set of values that are obtainable by synthesizing a term from the grammar,
again with respect to the original semantics $\sem{\cdot}$.
Because the set of values is overapproximated, a term synthesized
using the abstract semantics may not satisfy the specification
when executed with the standard semantics.
However, by showing the desired output is \emph{absent} from the set of obtainable values,
one can prove \emph{unrealizability} in a sound manner!

\begin{theorem}[Soundness of Abstract Semantics for Unrealizability]
  For a \semgus problem $sem = (G, \forall x. \psi(x, f(x)))$, if $sem$ is
  unrealizable when $G$ is equipped with an abstract semantics $\sem{\cdot}^{\#}$, then
  $sem$ is also unrealizable when $G$ is equipped with $\sem{\cdot}$.
\end{theorem}


Equipping a \semgus problem with an abstract semantics still results in a \semgus problem, 
which can be solved using the procedure described in $\S\ref{Se:encoding}$.
Much like how abstract semantics are used for efficient program verification,
an abstract semantics can sometimes be used to prove the unrealizability 
of a \semgus problem with the original semantics in a much more efficient manner.

\subsection{Solving Realizable \semgus Problems with Underapproximating Semantics}
\label{Se:semgus-underapproximate}

In this section, we show that an \emph{underapproximating semantics}, 
can be used to synthesize solutions to \emph{realizable} \semgus problems.

\begin{definition}
  For a grammar $G$ equipped with a semantics $\sem{\cdot}$, we say
  $\sem{\cdot}^{\flat}$ \emph{underapproximates} $\sem{\cdot}$ on $G$, or that
  $\sem{\cdot}^{\flat}$ is an \emph{underapproximating semantics} for $G$ with respect to
  $\sem{\cdot}$, if for every term $t \in L(G)$, every state $\Gamma$, and every value $v$ on which
  $\sem{\cdot}^{\flat}$ \emph{is defined},
  $\sem{t}^{\flat}(\Gamma, v) = \sem{t}(\Gamma, v)$.
\end{definition}

Intuitively, an underapproximating semantics is defined as a subset of the
original semantics.
Outside of the subset upon which it is defined, an underapproximating
semantics is undefined, which does not mean that a term can evaluate to any value,
but rather that a term \emph{cannot} evaluate to any value.
More precisely, one cannot prove any theorems about the relation
$\sem{t}^{\flat}(\Gamma, v)$ if $\sem{\cdot}^{\flat}$ is undefined on $t, \Gamma$, and $v$.
Instead, an underapproximate semantics is precise on the subset upon which it is defined,
i.e., $\sem{t}^{\flat}(\Gamma, v) = \sem{t}(\Gamma, v)$ if $\sem{\cdot}^{\flat}$ is defined on $t, \Gamma$, and $v$.

In \semgus, an underapproximating semantics corresponds to a problem where
synthesized terms only have meaning if their semantics is defined on the input-output examples.
For the subset of terms for which the semantics is defined,
the semantics is exact, which allows underapproximating semantics to be used for
program synthesis.
Because there may be an answer to the problem outside the defined subset,
an underapproximating semantics cannot be used for unrealizability.

\begin{theorem}[Soundness of Underapproximating Semantics for Synthesis]
   For a \semgus problem $sem = (G, \forall x. \psi(x, f(x)))$, if $sem$ is
   realizable with solution $t$ when $G$ is equipped with an underapproximating semantics
   $\sem{\cdot}^{\flat}$, then $t$ is also a solution for $sem$ when $G$ is equipped with $\sem{\cdot}$.
 \end{theorem}


An underapproximate semantics indirectly restricts
the search space for program synthesis.
This restriction is not necessarily related to the grammar supplied to a \semgus problem,
but may have a semantic meaning---for example, a bound on the number of possible loop iterations.

As is the case with an abstract semantics, \semgus can be supplied with an underapproximate 
semantics to yield a relatively more efficient procedure for program synthesis, as illustrated in \S\ref{Se:underapproximate-semantics}.



\section{Solving Semantics-Guided Synthesis Problems Via Constrained Horn Clauses}
\label{Se:encoding}

This section presents a general procedure for encoding \semgus problems so that they can be solved by 
answering a query over Constrained Horn Clauses, which in turn can be solved by an off-the-shelf
CHC solver.

\S\ref{Se:cegis} describes how general \semgus problems can be solved by solving 
\semgus-with-examples problems in tandem with
counterexample-guided inductive synthesis; it also states the correctness of our solving procedure.
\S\ref{Se:listing-sol} presents a method for using \emph{flattened representations} 
of terms as opposed to trees, to avoid the use of 
algebraic datatypes in SMT solvers.

\subsection{Solving \semgus Problems with Counterexample-Guided Inductive Synthesis}
\label{Se:cegis}

\emph{Counterexample-guided inductive synthesis} (CEGIS) is a widely implemented algorithm 
in program synthesizers.
The core idea of CEGIS is that instead of searching for a term that 
satisfies the specification for the entire input space, 
the synthesizer  searches for a solution that satisfies the specification 
on a finite set of examples $\examples$.
A verifier then attempts to prove that 
the solution is also correct on the universally quantified specification;
if not, a counterexample is added to the set of examples.
The algorithm then repeats.
The main advantage of CEGIS is that it eliminates the universal quantifier 
over the space of program inputs, yielding a simpler problem.

The algorithm sketched in \S\ref{Se:MotivatingExample}, as well as the one presented in 
\S\ref{Se:listing-sol}, is designed to solve 
\emph{\semgus-with-examples} problems, 
which are $\semgus$ problems where the specification is
given in terms of a set of examples $\examples$, and has the form
$\bigwedge_{x\in \examples}\psi(x, \sem{f}(x))$.
To solve general \semgus problems, the \semgus-with-examples 
algorithm can then be embedded within a CEGIS loop, where 
the specification is given in terms of the set of counterexamples accumulated by CEGIS.

The general idea of using CHCs to describe the syntax and semantics of a \semgus-with-examples
problem
$sem = (G, \forall x \in \examples.\psi(x, f(x)))$ has 
already been described in \S\ref{Se:MotivatingExample}: Equation~\eqref{eq:syntax_ex} and 
Equation~\eqref{eq:semantics_ex} show how the syntax and the semantics of a production 
$Start \rightarrow \Ewhile{B}{S}$ can be written as CHCs, and it is straightforward to 
describe other productions in this manner as well.

The final query that describes the specification can be formally written as the following rule.
\begin{equation}
  \label{eq:general-query}
\infer[\mathsf{Query}]{\mathit{Realizable}}
{
  \StartrreduceTree{t} \quad
  \bigwedge_{\ex_i \in \examples} \StartlreduceTree{\ex_i, t}{o_i} \quad 
  \bigwedge_{\ex_i \in \examples} \psi(\ex_i, o_i)
}
\end{equation}
$Realizable$ is the final theorem that shows whether the given \semgus-with-examples problem 
is realizable or not.
If the CHC solver finds a proof for $Realizable$, 
then the problem 
is \textit{realizable} and the program $t$
is a solution.
If the solver can establish that $Realizable$ is
\textit{unsatisfiable}, then the problem is \textit{unrealizable}.
The correctness of our algorithm can be stated as the following theorem:

\begin{theorem}[Soundness and Completeness] 
  \label{tree-correctness}
  Consider a \semgus-with-examples problem $sem = (G, \forall x \in \examples.\psi(x, f(x)))$, 
  equipped with semantic rules $\mathcal{R}_{sem}$, a specification set $\examples$, and the $\mathsf{Query}$ rule 
  (Equation~\eqref{eq:general-query}).
  Let the CHC form of $G$ be $\mathcal{R}_{syn}$.
  Then, \textit{Realizable} is a theorem over $\mathcal{R}_{sem}$ and $\mathcal{R}_{syn}$ 
  if and only if the \semgus-with-examples problem $sem$ is realizable.
  Moreover, if \textit{Realizable}  is a theorem, then the value of $t$ in the 
  $\mathsf{Query}$ rule satisfies $t \in L(G)$ and $\forall x \in \examples. \: \psi(x, \sem{t}(x))$.
\end{theorem}
Theorem~\ref{tree-correctness} can be proved by proving the correctness of the syntax rules via structural induction.

As shown in prior work~\cite{sketch}, the CEGIS algorithm is often 
powerful enough for program synthesis, where a term synthesized for the 
given examples generalizes to the entire space of possible inputs.
Prior work on unrealizability~\cite{unreal,semilinear} also shows that 
CEGIS is often powerful enough
to prove that a synthesis problem is unrealizable---i.e., the problem
does not admit a solution even when only a finite number of examples are considered.

\begin{example}
	The problem of synthesizing a program for bitwise-xor
  described in \S\ref{Se:MotivatingExample} 
  can be written as a \semgus-with-examples problem
  $sem = (G_{ex}, \forall_{x, y \in \examples_{ex}} f(x, y) = x \oplus y)$, 
  where $\examples_{ex} = [(6, 9), (44, 247), (14, 15)]$.
  As seen in \S\ref{Se:MotivatingExample}, $\examples_{ex}$ is sufficient to prove 
  that $sem$ is unrealizable.
\end{example}

In particular, for a \semgus problem $sem$, the CEGIS algorithm is sound but incomplete for unrealizability~\cite{unreal}.
As discussed in \S\ref{Se:Evaluation}, CEGIS is still able to synthesize solutions to, or prove unrealizability of, many \semgus problems.
However, this procedure is incomplete.
\begin{theorem}[CEGIS for unrealizability~\cite{unreal}]
  Let $sem_\examples$ be a \semgus-with-examples problem identical to $sem$, but where the specification is given 
  over the input examples $\examples$.
  If $sem_\examples$ is unrealizable, then $sem$ is unrealizable as well.
  However, there exists an unrealizable \semgus problem $sem$ for which $sem_\examples$ is realizable for any finite set of examples $\examples$.
\end{theorem}

\subsection{Using Flattened Representations of Terms to Solve \semgus Problems}
\label{Se:listing-sol}

While it is possible to solve \semgus-with-examples
problems using terms encoded as trees 
using the scheme given
in \S\ref{Se:trees-and-semantics}, 
current solvers sometimes fail to return an answer 
depending on how well they can handle trees encoded as algebraic datatypes.
In this section, we show how to alleviate this problem by using a 
\emph{flattened representation} of terms, which we refer to as a listing.\footnote{
  Listings may be implemented as lists or arrays in an SMT solver.
}
The idea is that a term $t$ can be encoded using a
\emph{pre-order listing} $\aI_t$ of the productions applied to derive $t$.

\begin{example}
  \label{ex:listings}
  Consider once more the term $t = \Ewhile{(x<x)}{\Eassign{x}{x}}$ from Example~\ref{ex:term-tree}, 
  constructed from the grammar $G_{ex}$ in Figure~\ref{fig:gex}.
  The pre-order listing of productions applied to derive $t$ is 
  $[\ballns{1},\ballns{2},\ballns{6},\ballns{6},\ballns{4},\ballns{6}]$, 
  where $Start \rightarrow \Ewhile{B}{S}  \ball{1}$ is the first production applied to the 
  nonterminal $Start$, the next production $B \rightarrow E < E  \ball{2}$ is applied next, 
  and the remaining productions are applied in left-to-right order as well.
\end{example}

Following the list representation of terms, 
the next step is to modify the syntax relations and rules to operate over lists.
Equation~\eqref{eq:syntax-list} describes the 
syntax rule generated using a flattened representation of terms for the 
production $A_0 \rightarrow \sigma(A_1, \cdots, A_i) \sball{n}$.
\begin{equation}
  \label{eq:syntax-list}
  \infer[syntax^\mathsf{List}_{A_0 \rightarrow \sigma(A_1, \cdots, A_i) \sball{n}}]{\mathsf{syn}_{A_0}(\aI_{in}, \sballlrns{n}::\aI_1)}
  {
  \mathsf{syn}_{A_i}(\aI_{in}, \aI_i) \quad 
  \mathsf{syn}_{A_{i-1}}(\aI_i, \aI_{i-1}) \cdots
  \mathsf{syn}_{A_1}(\aI_2, \aI_1)
  }
 \end{equation}

There are several things to notice about Equation~\eqref{eq:syntax-list}.
First, the syntax relation $\mathsf{syn}_N$ now ranges over two listings (term representations) 
as opposed to a single term, where the first listing may be interpreted as 
an incoming listing and the second an outgoing listing.
Here, the relations should evaluate to $\Etrue$ if and only if the outgoing 
listing is equivalent to the pre-order representation of the term concatenated to the 
incoming listing.

Second, the outgoing listing of a nonterminal is passed as the incoming 
listing of the next nonterminal in \emph{right-to-left order}, followed by 
prepending the number of the production to the head of the listing.
This algorithm effectively creates a pre-order representation of a term 
by performing a post-order traversal, appending each production encountered
to the head of the listing.

\begin{example}
  Consider Equation~\eqref{eq:syntax_ex} from \S\ref{Se:MotivatingExample}, which 
  describes the syntax rule for the production $\mathit{Start} \rightarrow \Ewhile{B}{S} \ball{1}$.
  Using a list representation of terms, the rule would be modified to:
  \begin{equation}
    \label{eq:syntax-list-example}
   \infer[syntax^{\mathsf{List}}_{\mathit{Start} \rightarrow \Ewhile{B}{S} \ball{1}}]
  {\Startrreduce{\aI_{in}}{\sballlrns{1}::\aI_1}}
  {
    \begin{array}{l}
    \SFrreduce{\aI_{in}}{\aI_2} \quad
    \BFrreduce{\aI_2}{\aI_1}
    \end{array}
  }
  \end{equation}
  Equation~\eqref{eq:syntax-list-example}
  traverses the nonterminals $B, S$ in right-to-left order, then 
  prepends the identifier $\ballns{1}$ to the head of the list $\aI_{1}$.
\end{example}

Having encoded a pre-order representation of a term, the semantic 
rules must interpret this representation accordingly as well.
The semantic relations now also range over 4 elements: an 
incoming listing $\aI_{in}$ and an incoming state $\Gamma$, followed 
by an outgoing listing $\aI_{out}$ and a resulting value $v$.
They should evaluate to $\Etrue$ if and only if for the list $\aI_t$
such that $\aI_{in} = \aI_t ++ \aI_{out}$, $\sem{t}(\Gamma, v)$ 
also evaluates to $\Etrue$ for the corresponding term $t$.

Keeping that in mind, a semantic rule that uses a flattened 
representation of terms for the production 
$A_0 \rightarrow \sigma(A_1, \cdots, A_i) \sball{n}$, equipped with the semantics 
$\phi \wedge \mathsf{sem}_{A_1}(\langle \Gamma_1, t_1 \rangle, v_1), \cdots, \mathsf{sem}_{A_i}(\langle \Gamma_i, t_i \rangle, v_i){\implies} \mathsf{sem}_{A_0}(\langle \Gamma, t \rangle, v_0)$ is described in 
Equation~\eqref{eq:semantics-list}.

\begin{equation}
  \label{eq:semantics-list}
  \infer[sem^{\mathsf{List}}_{A_0 \rightarrow \sigma(A_1, \cdots, A_i) \sball{n}}]
  {\mathsf{sem}_{A_0}(\langle \Gamma, \sballlrns{n}::\aL_{1} \rangle, \langle v_0, \aI_{out} \rangle)}
  {
    \phi \quad 
    \mathsf{sem}_{A_1}(\langle \Gamma_1, \aI_1 \rangle, \langle v_1, \aI_2 \rangle) \cdots 
    \mathsf{sem}_{A_i}(\langle \Gamma_i, \aI_i \rangle, \langle v_i, \aI_{out} \rangle)
  }
\end{equation}
Because the syntax rules have encoded terms as a pre-order listing, 
the semantic rules are free to interpret the current production 
by checking the head of the list, then compute values for subterms 
in \emph{left-to-right} order.
The actual semantics of the production remains encoded in $\phi$.

\begin{example}
  Consider Equation~\eqref{eq:semantics_ex} from \S\ref{Se:MotivatingExample}, which 
  describes the semantic rule for the production $\mathit{Start} \rightarrow \Ewhile{B}{S} \ball{1}$.
  Using a list representation of terms, the rule would be modified to:
  \begin{equation}
    \infer[]
		{\Startlreduce{\Gamma, \ballns{1}::\aI_1}{\Gamma_2, \aI_{out}}}
		{
		\begin{array}{l}
      \BFlreduce{\Gamma, \aI_1}{\Etrue, \aI_2} \ \ 
      \SFlreduce{\Gamma, \aI_2}{\Gamma_1, \aI_{out}} \ \ 
      \Startlreduce{\Gamma_1, \ballns{1}::\aI_1}{\Gamma_2, \aI_{out}}
		\end{array}
		}
  \end{equation}
	The list $\ballns{1}::\aI_{in}$ represents the entire term for $\Ewhile{B}{S}$ in 
  preorder---the tailing list $\aI_{out}$ represents the part that comes after $\Ewhile{B}{S}$.
\end{example}

\begin{theorem}[Correctness of Listings]
  \label{listing-correctness}
  Let $\mathcal{R}^{\mathsf{List}}_{sem}$ be a set of semantic rules using a flattened representation 
  of terms, created from the set of semantic rules $\mathcal{R}_{sem}$.
  For any nonterminal $N$, $\NFlreduce{\Gamma, \aL_{in}}{v, \aL_{out}}$ is a theorem of 
  $\mathcal{R}^{\mathsf{List}}_{sem}$ iff
  $\NFlreduceTree{\Gamma, t}{v}$ is a theorem of $\mathcal{R}_{sem}$, and $\aL_{in} =\aL_t \aL_{out}$ 
  (i.e., \texttt{concat}($\aL_t$, $\aL_{out}$)) where $\aL_t$ is the pre-order listing of a term 
  $t \in L(N)$.
\end{theorem}

Theorem~\ref{listing-correctness} states the correctness of the 
flattened term representations, and can be proved using induction on the height of the derivation tree.
The specification query is similar to 
the one given in Equation~\eqref{eq:general-query}, except that the new syntactic and semantic relations 
are used in place of the old ones.

\section{Instantiating \semgus with Various Semantics}
\label{Se:Semantics}

We now proceed to showcase the capabilities of the \semgus framework 
by instantiating it with a variety of semantics to solve imperative
program-synthesis problems.
\cbstart
In \S\ref{Se:vectorized-semantics}, \S\ref{Se:abstract-semantics}, and 
\S\ref{Se:underapproximate-semantics}, we are concerned with various different 
semantics for the imperative programming language $G_{impv}$, from 
Figure~\ref{fig:grm}.
\cbend
Values in $G_{impv}$ range over integers, bitvectors, Boolean values, and arrays.
$G_{impv}$ contains most common imperative structures, such as
assignments, branches and loops.
Imperative grammars that use the same operators but different productions can be 
viewed as being derived from $G_{impv}$, which means that the techniques introduced in
this section are applicable to \emph{any} imperative grammar, as long as they 
use a subset of the operators in $G_{impv}$.

\begin{figure}[t!]
  {
$$
\begin{array}{llcl}
  \mathit{Stmt} & S & ::= &
        \Eassign{x}{E} \mid \Eassign{x}{C} \mid \Earrayupdate{arr}{E}{E} \mid \Eseq{S}{S} \ | \
        \Eifthenelse{B}{S}{S}  \mid \Ewhile{B}{S}\\
  \mathit{BVExpr} & C & ::= &
        x  \mid \bar{0} \mid \bar{1} \mid C~\&~C \mid (C~|~C) \mid\ !C \mid C+C \mid \Eifthenelse{B}{C}{C} \\        
  \mathit{IntExpr} & E & ::= &
        x  \mid 0 \mid 1 \mid x \mid E + E \mid  \Eifthenelse{B}{E}{E} \mid \Earrayaccess{arr}{E} \\
  \mathit{BoolExpr} & B & ::= &
        \Etrue \mid \Efalse \mid \lnot B \mid B \wedge B  \mid E < E  \mid C < C
\end{array}
$$
  }
\caption{The general imperative grammar $G_{\mathit{Impv}}$ that we are interested in.} \label{fig:grm}
\end{figure}

In \S\ref{Se:vectorized-semantics}, we discuss how to instantiate an imperative \semgus 
problem with an alternative exact semantics.
This semantics, called a \emph{vectorized semantics}, sidesteps the 
problem of having to consider 
multiple examples separately.
In \S\ref{Se:abstract-semantics}, we show how \semgus can be instantiated with 
an \emph{abstract semantics} to prove the unrealizablity of a synthesis problem, and 
in \S\ref{Se:underapproximate-semantics}, how an \emph{underapproximating semantics} can 
be used to more efficiently compute solutions for a realizable problem.

\cbstart
Finally, in \S\ref{Se:regex-semantics}, we sketch how \semgus can be
instantiated with a semantics for regular expressions to create a tool
to synthesize regexes from positive and negative examples.
\cbend

\subsection{Instantiating \semgus with an Alternative Exact Semantics}
\label{Se:vectorized-semantics}

A straightforward way of instantiating a \semgus problem is to supply 
$\semgus$ with a standard semantics, as discussed in \S\ref{Se:MotivatingExample} 
and \S\ref{Se:encoding}.
For example, the three rules in Figure \ref{fig:vector-sem-top} are 
standard semantic rules that define the semantics of the terms 
``$\Eassign{x}{e}$'' and ``$\Ewhile{b}{s}$''.
These semantics operate over a 
single state, and compute exact values for all terms in the program.

However, this straightforward approach induces a substantial drawback in the $\mathsf{Query}$ rule
in Equation~\eqref{eq:general-query}.
In each premise of the $\mathsf{Query}$ rule, the solver must re-derive 
proof trees for each example, even though they are all structurally similar
due to sharing the same term representation.

To mitigate this inefficiency, we develop
a different exact semantics, called the \emph{vectorized semantics}, and show 
that \semgus can be instantiated with this semantics as well.
The vectorized semantics modifies the semantics of standard imperative programs 
to accommodate and execute multiple examples simultaneously in the form of vectors.
This idea allows us to merge the examples, as well as the semantic premises 
$\NFlreduceTree{\aT, \ex_1}{o_1}, \cdots, \NFlreduceTree{\aT, \ex_n}{o_n}$ of the $\mathsf{Query}$ 
rule, into a single semantic premise 
$\NFlreduceTree{\aT, \Vec{\ex}}{\Vec{o}}$, where 
$\Vec{\ex}$ and $\Vec{o}$ represent the vectorized input-output examples.

The main challenge in defining a vectorized semantics is that, in the 
presence of loops and conditionals, different examples can cause a given loop 
to run a different number of times, and can take different branches of an 
if-statement.
Here, we note that \semgus is not the cause of these challenges, nor does 
it require the vectorized semantics; rather, \semgus is what 
provides us with the 
possibility of defining different semantics that are better suited
to solving the task at hand.

\begin{figure}[t!]
  \begin{subfigure}{\textwidth}
    {\[
\begin{array}{c}
\infer[\mathsf{Assign}]{\Sreduce{\Eassign{x}{e}}{\Gamma}{\Gamma_r}}
  {\Ereduce{e}{\Gamma}{v} \quad 
  \Gamma_r=\Gamma[x \mapsto v]}
  \quad
\infer[\mathsf{WFalse}]{\Sreduce{\Ewhile{b}{s}}{\Gamma}{\Gamma}}
  {
    \Ereduce{b}{\Gamma}{v_b} \quad v_b = \Efalse
  }
\\[2mm]
\infer[\mathsf{WTrue}]
  {\Sreduce{\Ewhile{b}{s}}{\Gamma}{\Gamma_2}}
  {
    \Ereduce{b}{\Gamma}{v_b} \quad v_b = \Etrue \quad \Sreduce{s}{\Gamma}{\Gamma_1} \quad
    \Sreduce{\Ewhile{b}{s}}{\Gamma_1}{\Gamma_2}
  }
  \end{array}
\]}
  \caption{Standard semantic rules for the terms $\Eassign{x}{e}$ and $\Ewhile{b}{s}$ in $L(G_{\mathit{Impv}})$, 
    where the semantic function is denoted by $\sem{\cdot}$.}
  \label{fig:vector-sem-top}
  \end{subfigure} \newline
\begin{subfigure}{\textwidth}
  {\[
\begin{array}{c}
\infer[\mathsf{Assign_\examples}]{\eSreduce{\Eassign{x}{e}}{\VG}{\VG_r}}
  {\eEreduce{e}{\VG}{\Vv} \quad 
\forall i.\ 
  \VG_r[i]=\VG[i][x \mapsto \Vv[i]]}
  \quad
\infer[\mathsf{WFalse_\examples}]{\eSreduce{\Ewhile{b}{s}}{\VG}{\VG}}
  {
    \eEreduce{b}{\VG}{v_b} \quad \forall i.\ \Vv_b[i] = \Efalse
  }
\\[2mm]
\infer[\mathsf{WTrue_\examples}]
  {\eSreduce{\Ewhile{b}{s}}{\VG}{\Vjoin{\Vproj{\VG}{\lnot \Vv_b}}{\Vproj{\VG_2}{\Vv_b}}}}
  {
    \eEreduce{b}{\VG}{\Vv_b} \quad \exists i.\ \Vv_b[i] = \Etrue \quad \eSreduce{s}{\Vproj{\VG}{\Vv_b}}{\VG_1} \quad
    \eSreduce{\Ewhile{b}{s}}{\VG_1}{\VG_2}
  }
  \\[2mm]
  \Vproj{\VG}{\Vv_b} \defeq [\Eifthenelse{\Vv_b[0]}{\VG[0]}{\bot}$, $\cdots$, $\Eifthenelse{\Vv_b[n-1]}{\VG[n-1]}{\bot}]
  \\[2mm]
  \Vjoin{\Vproj{\VG}{\lnot \Vv_b}}{\Vproj{\VG'}{\Vv_b}} \defeq
  \\[1mm]
  [\Eifthenelse{\Vv_b[0]}{\VG[0]}{\VG'[0]}, \cdots, \Eifthenelse{\Vv_b[n-1]}{\VG[n-1]}{\VG'[n-1]}]
  \end{array}
\]}
  \caption{Sample vectorized semantic rules for the terms $\Eassign{x}{e}$ and $\Ewhile{b}{s}$ in $L(G_{\mathit{Impv}})$, 
  where the (vectorized) semantic function is denoted by $\sem{\cdot}_{\examples}$.
  $\bot$ is a special state that ignores all computation performed.}
  \label{fig:vector-sem-bottom}
\end{subfigure}
  \caption{Standard and vectorized semantics for the terms $\Eassign{x}{e}$ and $\Ewhile{b}{s}$.
  }
\label{fig:vector-sem}
\end{figure}

The three rules in Figure~\ref{fig:vector-sem-bottom} present the big-step semantics for the terms 
$\Eassign{x}{e}$ and $\Ewhile{b}{s}$, the terms 
that are most relevant to overcoming these challenges.
The most interesting rule here is $\mathsf{WTrue_\examples}$.
This rule states that as long as one of the examples in the vector makes the guard 
$b$ true, the body of the loop should be entered.
However, only the variable valuations that make the guard true are updated
in the loop-body $s$ (the $\Vproj{\VG}{\Vv_b}$ operator sets all valuations for which
the guard is false to the special value $\bot$).
The whole process is repeated (using the projected vector of valuations) until
all entries of $\Vv_b$ are $\bot$, as stated in $\mathsf{WFalse_{\examples}}$.
Finally, the vector of valuations in the bottom of the rule contains the \textsc{merge} of
valuations for which the guard was false, and valuations $\VG_2$ that resulted from
running the loop on the valuations $\Vproj{\VG}{\Vv_b}$ for which the guard was true.

The following theorem states the correctness of the vectorized semantics.

\begin{theorem}[Correctness of Vectorized Semantics]
  \label{vectorized-correctness}
Given a set of examples $\examples=[\Gamma_1,\ldots, \Gamma_n]$ and a term $t$,
$\sem{t}_\examples([\Gamma_1,\ldots, \Gamma_n], [\Gamma_1',\ldots, \Gamma_n'])$
if and only if for every $1\leq i\leq n$, we have
$\sem{t}(\Gamma_i,\Gamma_i')$.
\end{theorem}

When supplying vectorized semantics to a \semgus-with-examples problem, one
should supply a \emph{single} vectorized example that contains all the examples 
from the original example set.
Aside from this difference in how examples should be supplied, the vectorized 
semantics can be treated just like any other semantics, meaning that the CHC-based 
solving procedure from \S\ref{Se:MotivatingExample} and \S\ref{Se:encoding} still 
holds.
Moreover, as stated at the start of this section, 
the vectorized semantics illustrated above can be generated automatically for all 
subgrammars of $G_{impv}$, which allows it to be used as a general optimization 
for solving imperative \semgus problems (as our tool \ourtool does).

\subsection{Using Abstract Semantics in \semgus to Prove Unrealizability}
\label{Se:abstract-semantics}

In this section, we show how the grammar $G_{impv}$ can be instantiated 
with an abstract semantics to prove the unrealizability of \semgus problems, 
following the idea introduced in \S\ref{Se:semgus-abstract}.

There are many abstract semantics with which one can equip a language.
Here, we use the abstract domain $\mathbb{B}_{i}$
presented in \S\ref{Se:semgus-other} 
as an example, which tracks only the $i$-th bit of 
a variable using three values: $\Etrue, \Efalse$ and $\top$ (the join of $\Etrue$ and $\Efalse$).

\begin{example}
	Recall Equation~\eqref{eq:abstract-and}, which represents the abstract semantics for a term 
  $e_1~\&~e_2$ from $G_{ex}$ of \S\ref{Se:MotivatingExample}, using the abstract domain $\mathbb{B}_3$.
  The right-hand side of the final premise describes the abstract semantic function $\sem{\&}^{\#}$ for the 
  operator $\&$, which sends the computation to $\top$ if any of $v_1^\#$ or $v_2^\#$ are $\top$, 
  and computes the exact value otherwise.
  Note how the semantic relations, as well as the structure of the semantic rule, remain unchanged---from 
  the viewpoint of $\semgus$, an abstract semantics expressed using CHCs is merely a different semantics supplied to \semgus, 
  for which one can apply the same solving procedure as given in \S\ref{Se:MotivatingExample} and \S\ref{Se:encoding}.
\end{example}

Different abstract domains have different degrees of efficiency and precision 
in \semgus.
To see why, consider how one would deal with branches using the 
abstract domain described above.
This particular abstract domain cannot handle comparisons well because it 
only tracks a single bit, and thus 
it is almost always the case that one does not know which branch to take
in an if-statement.
There are two possible approaches in this situation---one may just choose to assign 
$\top$ to the result of the branch, or one may try and execute 
both branches and assign their join to the result.
This problem arises for both if-then-else statements and loops.
As an example, two different rules for loop iteration
are described in Example~\ref{ex:abstract-while}.

\begin{example}
  \label{ex:abstract-while}
  Equations~\eqref{eq:abstract-while-imprecise} and~\eqref{eq:abstract-while-precise} present
  different abstract semantics for the term $\Ewhile{b}{s}$ from 
  $G_{ex}$ of \S\ref{Se:MotivatingExample}, 
  using the abstract domain $\mathbb{B}_3$, which tracks only the third bit of each variable.

	\begin{equation}
	\label{eq:abstract-while-imprecise}
    {\infer[\rWT_{\mathsf{Havoc}}^{\#}]
	{\sem{\Ewhile{b}{s}}^\#(\Gamma^\#,\Gamma^\#_r)}
	{
		\sem{b}^\#(\Gamma^\#,v^\#_b)
		\quad
		\sem{s}^\#(\Gamma^\#, \Gamma^\#_1)
		\quad
		\sem{\Ewhile{b}{s}}^\#(\Gamma^\#_1, \Gamma^\#_2)
		\quad \Gamma^\#_r = \top
  }}
 \end{equation}
	\begin{equation}
	\label{eq:abstract-while-precise}
    {\infer[\rWT_{\mathsf{Join}}^{\#}]
	{\sem{\Ewhile{b}{s}}^\#(\Gamma^\#,\Gamma^\#_r)}
	{
		\sem{b}^\#(\Gamma^\#,v^\#_b)
		\quad
		\sem{s}^\#(\Gamma^\#, \Gamma^\#_1)
		\quad
		\sem{\Ewhile{b}{s}}^\#(\Gamma^\#_1, \Gamma^\#_2)
		\quad \Gamma^\#_r = \textsc{join}(\Gamma^\#, \Gamma^\#_2)
    }}
	\end{equation}
  In both scenarios, the value of $v_b$ will be $\top$ because knowing
  only the third bit does not give us enough information to resolve 
  a condition of the term ``$e < e$'' from $G_{ex}$.
	In this situation, the rule $\rWT_{\mathsf{Havoc}}$ simply gives up and assigns 
  $\top$ to the resulting value $\Gamma_r$.
  On the other hand,the rule $\rWT_{\mathsf{Join}}$ attempts to preserve some precision 
  by assigning the join of when the condition evaluates to $\Etrue$ ($\Gamma_2$, as the loop
  iterates in this case) and when the condition evaluates to $\Efalse$ ($\Gamma$, as the loop 
  body does not execute).
  If both $\Gamma$ and $\Gamma_2$ contain $x^\# = \Etrue$, then 
  $\rWT_{\mathsf{join}}$ is capable of inferring that the result of 
  $\Ewhile{b}{s}$ also has $x^\# = \Etrue$, while $\rWT_{\mathsf{Havoc}}$ cannot. 
\end{example}

The semantics expressed by $\rWT_{\mathsf{Join}}$ 
is more precise and more expensive than the first option.
For the example in \S\ref{Se:MotivatingExample}, an abstract semantics 
using $\rWT_{\mathsf{Havoc}}$ will fail to prove unrealizability of synthesizing 
bitwise-xor, because it cannot resolve the branch of the loop.
On the other hand, the added precision from  $\rWT_{\mathsf{Join}}$ succeeds in 
proving unrealizability, showing how different abstract domains 
can solve different $\semgus$ problems.

\cbstart
\jinwoo{Review B: Discussion on the selection of an abstract domain}
In general, there are many different abstract domains that one could use for a \semgus problem, 
as well as automated methods~\cite{learning_abs_synth, abs_synth} to discover them.
An interesting line of future work would be to design an algorithm for extracting abstract domains 
from \semgus proofs (like we did in our selection of the domain $\mathbb{B}_i$)
to solve other \semgus problems more efficiently.
\cbend

\subsection{Using Underapproximating Semantics in \semgus for Program Synthesis}
\label{Se:underapproximate-semantics}

In this section, we demonstrate how $\semgus$ can be equipped with an
\emph{underapproximating semantics} to perform program synthesis, following the 
idea from \S\ref{Se:semgus-underapproximate}.
Example~\ref{ex:under-example} shows an underapproximating semantics that sets a bound on the number of times 
each loop may be executed, as in bounded model checking~\cite{cbmc}.
The change to the semantics is simple---one simply adds a bound to the state 
and decreases the bound by one each time a loop iteration is performed.

\begin{example} 
  \label{ex:under-example}
  Equations~\eqref{eq:while-true-bound} and~\eqref{eq:while-false-bound} present 
  an underapproximating semantics for the term $\Ewhile{b}{s}$, 
  where the number of loop iterations is bounded by a fresh variable $i$.
  \begin{equation}
  \label{eq:while-true-bound}
    {\infer[\mathsf{WTrue}^{\flat}]
    {\sem{\Ewhile{b}{s}}^{\flat}(\langle \Gamma, i \rangle, \langle \Gamma_r, i \rangle)}
  {
    \sem{b}^{\flat}(\langle \Gamma, i \rangle, \Etrue)
    \quad 
    i > 0 \quad 
    \sem{s}^{\flat}(\langle \Gamma, i \rangle, \langle \Gamma', i \rangle)
    \quad
    \sem{\Ewhile{b}{s}}^{\flat}(\langle \Gamma', i{-}1 \rangle, \langle \Gamma_r, i{-}1 \rangle)
  }}
  \end{equation}
  \begin{equation}
  \label{eq:while-false-bound}
    {\infer[\mathsf{WFalse}^{\flat}]
    {\sem{\Ewhile{b}{s}}^{\flat}(\langle \Gamma, i \rangle, \langle \Gamma, i \rangle)}
  {
    \sem{b}^{\flat}(\langle \Gamma, i \rangle, \Efalse)
    \quad 
    i > 0
  }}
  \end{equation}
  
  One can see how these rules are underapproximating by 
  considering why one is unable to build a proof tree for a loop that must execute more iterations than the unrolling bound.
  For example, let the unrolling bound be $i = 1$.
  To prove that $\sem{\Ewhile{b}{s}}^{\flat}(\langle \Gamma, 1 \rangle, \langle \Gamma_r, 1 \rangle)$, i.e., the 
  conclusion with $i = 1$, 
  one would also require a proof for the final premise in the rule, namely 
  $\sem{\Ewhile{b}{s}}^{\flat}(\langle \Gamma',  0 \rangle, \langle \Gamma_r, 0 \rangle)$.
  However, a proof of $\sem{\Ewhile{b}{s}}^{\flat}(\langle \Gamma',  0 \rangle, \langle \Gamma_r, 0 \rangle)$
  requires that $0 > 0$ due to the third premise $i > 0$, which is unsatisfiable.
  Thus, nothing can be proved about $\sem{\Ewhile{b}{s}}^{\flat}(\langle \Gamma',  0 \rangle, \langle \Gamma_r, 0 \rangle)$---corresponding 
  to the fact that $\sem{\Ewhile{b}{s}}^{\flat}(\langle \Gamma',  0 \rangle, \langle \Gamma_r, 0 \rangle)$, and any 
  relations that rely on this premise, are \emph{undefined}.

  In contrast, the semantics described by Equation~\eqref{eq:while-true-bound} match exactly the standard 
  semantics of a while loop for a loop that executes fewer iterations than the unrolling bound.
\end{example}

The constraints that make a semantics underapproximating---for 
example, $i > 0$ in Example~\ref{ex:under-example}---can be encoded in the constraint element 
$\phi$ of a CHC.

\subsection{Instantiating \semgus to Synthesize Regular Expressions}
\label{Se:regex-semantics}

\cbstart
\jinwoo{Shepherd Item 4: An additional domain solvable by \semgus}
In this section, we move away from imperative programs and show how one
can use the \semgus framework to solve problems from a different
domain---namely, regular expressions.
For convenience, in this section we revert to expressing semantics as CHCs
(as done in \S\ref{Se:MotivatingExample}).
We assume a standard grammar for regular expressions.
{
$$
  \mathit{Regex} \ni \ R \ ::= \ \ c \quad \mid \quad \epsilon \quad \mid \quad \phi \quad \mid \quad (R \mid R) \quad \mid \quad R \cdot R \quad \mid \quad R^* 
$$
}
Encoding semantics for regular expressions via CHCs poses an interesting challenge, in that 
the semantics often involves nondeterminism, either when dealing with Kleene star, or finding matching substrings 
for concatenation.
Generally, nondeterminism is naturally dealt with when using only positive examples, where, for example, 
the \emph{existence} of a proof of the relation 
$\RFlreduceTree{\omega, r}{\Etrue}$, where $\omega$ is the string to match and $r$ a regex, 
ensures that there \emph{exists} a run of $r$ that accepts $\omega$.
However, nondeterminism mixes poorly with negative examples, which specify strings that a regex should reject: the existence of a proof of the relation 
$\RFlreduceTree{\omega, r}{\Efalse}$ merely states that there exists a run of $r$ that rejects $\omega$, whereas the semantics of 
regular expressions dictate that \emph{no} run of $r$ should accept $\omega$ for $\omega$ to be rejected.


Fortunately, the expressiveness of CHCs and \semgus allows us to
develop an alternative, deterministic semantics for regexes.
Given an input string $\omega$ of length $n$, the semantics is expressed
in terms of $(n + 1)$-by-$(n + 1)$ upper-triangular matrices of
Boolean variables:
in matrix $X$, an element $X[i, j]$ indicates whether the considered regex
matches the substring $\omega[i, j]$ (as presented in~\cite{rfixer}).
Then, given regexes $r_1$ and $r_2$ that return the matrices $X_1$
and $X_2$, respectively,
the semantics of concatenation is Boolean matrix multiplication---i.e.,
$X_1 \cdot X_2$---where the element-multiplication operation is logical-and,
and element-addition is logical-or.
For the regex $r^*$, the entire set of substrings of $\omega$ matched
by $r^*$ can be computed by taking $X^0 + X^1 + \cdots X^n$, where $X$
is the matrix for substring acceptance by $r$, and the $+$ operator
denotes pointwise logical-or of the operand matrices.
The $+$ operator is also the interpretation of alternation in
regular expressions.


This approach gives us a deterministic semantics for regular expressions, but another interesting 
challenge lies in how these semantics should be concretely embedded in an SMT solver.
One could simply choose to encode the semantics directly using the theory of strings; some SMT solvers also directly 
support regular expressions as part of their theories.
However, it is also the case that strings are poorly supported at best, especially in CHC solvers---in our experiments, 
Z3 would often throw segmentation faults when asked to solve CHCs that contained strings.

For this reason, we again exploit the generality of \semgus, and develop a more solver-friendly semantics by encoding strings as sequences of integers.
The base cases for regular expressions---matching single characters, $\epsilon$, and $\phi$---can be encoded using equality between integers, 
and the rest of the semantics is compositional.
Figure~\ref{fig:regex_final} shows the final rule for the term $r^*$: the input string is encoded as a sequence of integers 
$s_0, s_1, \cdots, s_n$. 
The output is a matrix of Boolean variables $X$ as described above: 
the output argument of $\mathsf{Star_{\epsilon}}$ shows the substring matrix that matches $\epsilon$ (i.e., the identity matrix $I$).
The input parameter $k$ in $\mathsf{Star}$ iterates through $\mathsf{Star}$ to compute $X^k$,
along with the outputs of $\mathsf{sem}_{R}$ which show how concatenation is computed by matrix multiplication.
A regular expression $r$ accepts $\omega$ iff in $X_r$, the matrix for $r$, $X_r[0,n] = \Etrue$.

\begin{figure}
{\footnotesize
    $$
      \infer[\mathsf{Star_{\epsilon}}]{\RFlreduceTree{(s_0, s_1, \cdots, s_n), 0, r^*}{I}}{
      }
    $$

    $$
      \infer[\mathsf{Star}]{\RFlreduceTree{(s_0, s_1, \cdots, s_n), k, r^*}
        {X \cdot X_{k-1} + X_{k-1}}}{
        \RFlreduceTree{(s_0, s_1, \cdots, s_n), k, r}{X} \quad 
        \RFlreduceTree{(s_0, s_1, \cdots, s_n), k - 1, r^*}{X_{k - 1}}
      }
    $$
}
\vspace{-4mm}
  \caption{Semantics for the regular expression $r^*$, where $(s_0, s_1, \cdots s_n)$ represents a sequence of integers that encodes the input string.
} \label{fig:regex_final}
\end{figure} 
\cbend



\section{Implemention and Optimizations}
\label{Se:Implementation}

In this section, we describe our implementation of \ourtool, 
a solver for \semgus problems, as well as some optimizations 
that were applied in \ourtool.

\subsection{Implementing \ourtool}

At a high level, \ourtool accepts \semgus problems and 
encodes them as CHCs using the encoding in $\S\ref{Se:encoding}$.
It then passes the CHCs to Z3~\cite{z3}, which performs the actual proof search 
and produces an answer.
The output from Z3 is either UNSAT, which means the problem is unrealizable, or a proof for
$Realizable$ from Equation~\eqref{eq:general-query} using the inference rules from the \semgus problem: in this case, 
\ourtool can extract a solution to the \semgus problem from the proof.\footnote{
  Z3 may also time out, or produce an error for various reasons, for example when dealing with algebraic datatypes.
}
We note that the capability of \ourtool to synthesize programs also allows 
it to perform CEGIS for both program synthesis and unrealizability, which is unsupported in 
previous work on proving unrealizability~\cite{unreal, semilinear}.

We report here that Z3 itself varied in performance depending on whether 
particular internal flags were enabled.\footnote{
  The particular flags are $\mathsf{fp.xform.slice}, \mathsf{fp.xform.inline\_linear}$, and $\mathsf{fp.xform.inline\_eager}$.
}
While enabling these flags are the default setting for Z3 and result in better performance, 
they also made it difficult to recover the term representation from the output of Z3 
(which is required to synthesize a term).
Thus, during our evaluation in \S\ref{Se:Evaluation}, we disabled the flags; \ourtool can also 
be configured to run with the flags enabled.


%

In \S\ref{Se:encoding}, we looked at different ways of translating  \semgus 
problems into CHCs depending on whether trees or listings are used to 
represent terms.
\ourtool supports three configurations for representing terms---a configuration that uses 
algebraic datatypes to model trees, and two configurations that respectively use 
lists and arrays to encode listings.
In addition, \ourtool also implements a $\semgus$-specific optimization 
called the \emph{fused semantics} , described in \S\ref{Se:fused-semantics}.
For regular expressions, \ourtool simply implements the semantics in \S\ref{Se:regex-semantics} 
with terms encoded as arrays.

\cbstart
\jinwoo{Shepherd item 5: Enumerative Solver}
We also implemented an enumerative solver that uses the same CHC formalism, 
but replaces the syntax relations with concrete term representations (encoded using arrays) 
instead. This  solver enumerates terms and checks whether they are correct using the given semantics---i.e.,
for each enumerated term the CHCs are then again passed to Z3, which will verify whether the given term is correct or not.
The enumerative solver can be treated as a baseline alternative solver for \semgus problems that 
plugs into \ourtool (which is responsible for generating CHC scripts for Z3 to solve).

\jinwoo{Shepherd item 2: What does it take to implement a new \semgus domain?}
As presented throughout our paper, \ourtool supports various \semgus problems defined over many different 
domains and term representations.
Implementing a new domain for \semgus problems consists of two steps: 
a theoretical step where one develops a semantics for the domain on paper, then an implementation step 
to actually implement the domain.
The generality of \semgus often results in the first step requiring much more thought and effort: as shown in \S\ref{Se:Evaluation}, 
different ways of encoding semantics can lead to big differences in performance, 
and one must also consider the correctness of the semantics with respect to CHC solving,
as discussed in \S\ref{Se:regex-semantics}.

On the other hand, once the first step is completed, or if one is able to 
use a standard semantics, expressing the semantics as concrete semantics is a routine task.
For example, the vectorized semantics detailed in \S\ref{Se:vectorized-semantics} totals
around 350 lines of code in our implementation;
the regex semantics from \S\ref{Se:regex-semantics} totals around 100 lines,\footnote{
  The semantics were written in a DSL we designed to interface with the SMT-LIB CHC format for easy development.
}
which mostly implement pattern matching on terms and encoding SMT queries.
Moreover, given a library of various semantics for operators, defining a new \semgus problem is just as easy as 
defining a \sygus problem.
\cbend

\subsection{Optimizing Imperative \semgus Problems with Fused Semantics}
\label{Se:fused-semantics}

$\ourtool$ offers an optimization that utilizes a slightly different method of encoding 
syntax and semantic rules: instead
of building a term using the syntax rules and propagating 
it through the semantic rules separately, one can also think of a scheme
where the semantics of a term is executed on-the-fly while the
term is being constructed.
We refer to this kind of encoding as the \emph{fused semantics}.
Fused semantics are \emph{different} from supplying \semgus with a different semantics, 
because they are \emph{derived} from an original semantics that \semgus 
is supplied with.
Instead, one may think of them as an \emph{optimization} for \semgus problems 
over subgrammars of $G_{impv}$.

The key idea for fused semantics is to modify the syntax relations
so that they can check semantics as well as the syntactic structure,
and modify the syntax rules accordingly as well.
Thus, a syntax relation is now defined over three inputs---a term $t$, an input state $\Gamma$, and an output value $v$.
The relation should evaluate to $\Etrue$ if and only if
$t$ is a valid term, and $\sem{t}(\Gamma, v)$ is also $\Etrue$.
Generally, the syntax rule for a production
$A_0 \rightarrow \sigma(A_1, \cdots, A_i) \sball{n}$, again equipped with the semantics 
$\phi \wedge \mathsf{sem}_{A_1}(\langle \Gamma_1, t_1 \rangle, v_1), \cdots, \mathsf{sem}_{A_i}(\langle \Gamma_i, t_i \rangle, v_i)
\implies \mathsf{sem}_{A_0}(\langle \Gamma, t \rangle, v_0)$, 
can be generated in the
form of Equation~\eqref{eq:fused-rule}:
\cbstart
the structure of Equation~\eqref{eq:fused-rule} matches \emph{exactly} the structure of the supplied semantics 
(using the listing representation of terms).
\begin{equation}
  \label{eq:fused-rule}
  \infer[syntax^{\mathsf{fused}}_{A_0 \rightarrow \sigma(A_1, \cdots, A_i) \sball{n}}]
  {\mathsf{syn}^{\mathsf{fused}}_{A_0}(\langle \Gamma, \aL_{in} \rangle, \langle v, \sballlrns{n}::\aL_{out} \rangle)}
  {
    \phi \quad
    \mathsf{syn}^{\mathsf{fused}}_{A_1}(\langle \Gamma_1, \aL_{in} \rangle, \langle v_1, \aL_{mid} \rangle) \cdots 
    \mathsf{syn}^{\mathsf{fused}}_{A_i}(\langle \Gamma_i, \aL_{mid} \rangle, \langle v_i, \aL_{out} \rangle)
  }
\end{equation}

\jinwoo{Review B, E: Fused semantics and syntax relations}
One may ask how the fused semantics is different from a set of CHCs where the syntax relations have been eliminated, and the 
ordinary semantic rules implicitly check the syntax of a term while computing its semantics as well.
They are indeed similar, but there is a subtle difference in the order of the arguments: the semantics from 
Equation~\eqref{eq:fused-rule} have $\aL_{in}$ as the first argument, and $\aL_{out}$ as the second; whereas 
ordinary semantic rules have the list order reversed.\footnote{
  This subtle difference also has a difference in performance---syntax-relation-less CHCs perform similarly to 
  \semgus problems with syntax relations, which perform worse compared to the fused semantics as shown in \S\ref{Se:Evaluation}.
}
Intuitively, the specification of the fused semantics constructs a term while
computing the semantics on the fly, while the specification of the
ordinary semantics can only compute the semantics for a fully
constructed term.
\cbend

The new encoding presented in Equation~\eqref{eq:fused-rule} is enough to allow 
\emph{only the syntax rules} to describe both the syntax and semantics of terms within a \semgus problem, 
provided that the grammar does not contain productions with while loops.
However, productions that contain loops, such as $N \rightarrow \Ewhile{B}{S} \sball{n}$, 
require a separate procedure because
there must be a guarantee that the \emph{same loop body} is synthesized for each iteration.
To ensure that the same loop body is synthesized, one can either impose an additional constraint that states that
each synthesized loop body must be identical, or more simply, one can apply the
semantic relations described from $\S\ref{Se:encoding}$ instead.
\begin{equation}
\label{eq:fused-loop}
  \infer[syntax^{\mathsf{fused}}_{N \rightarrow \Ewhile{B}{S} \sballns{n}}]{\NFrreduceFused{\langle \Gamma, \Ewhile{b}{s} \rangle}{\Gamma_2}}
{
 \begin{array}{l}
  \BFrreduceFused{\langle \Gamma, b \rangle}{\Etrue} \quad
  \SFrreduceFused{\langle \Gamma, s \rangle}{\Gamma_1} \quad
   \NFlreduceTree{\Gamma_1, \Ewhile{b}{s}}{\Gamma_2}
 \end{array}
}
\end{equation}
Consider the rule 
given in Equation~\eqref{eq:fused-loop}.
Note that the first two relations from the premise are syntax relations that
both synthesize a term and execute its semantics.
In contrast, the third relation is a semantic relation, which is defined identically
to the semantic relations in $\S\ref{Se:MotivatingExample}$.
The semantic relations do not suffer from the problem of having to synthesize the
same loop body over multiple iterations.
The idea here is that the syntax relations synthesize the loop body on the first iteration,
then pass the representation to the semantic relations for subsequent iterations.

Finally, multiple $\mathsf{sem}_N$ premises
in the $\mathsf{Query}$ rule must be rewritten as $\mathsf{syn}^{\mathsf{fused}}_N$ as well;
when using non-vectorized semantics, this approach raises the same problem of potentially synthesizing different
\cbstart
\jinwoo{Review E: can the fused semantics be used with non-vectorized semantics as well?}
solutions for each example, and necessitates a constraint 
to ensure that all generated representations are identical.
Although it is possible to use fused semantics for non-vectorized semantics in this way,
$\ourtool$ employs the fused semantics as an optimization
to vectorized semantics only, which does not suffer from this problem,
because there is only a single vector of examples.
\cbend

\begin{theorem}[Soundness and Completeness of Fused Semantics]
   \label{fused-correctness}
   Let $\mathcal{R}_{syn}^{\mathsf{fused}}$ denote a set of fused syntax rules created according to the fused-semantic optimization
   described in \S\ref{Se:fused-semantics}, from a grammar $G$ equipped with a semantics $\mathcal{R}_{sem}$.
   Then for any nonterminal $N \in G$, $\NFrreduceFused{\langle \Gamma, t \rangle}{v}$ is a theorem over
  $\mathcal{R}_{syn}^{\mathsf{fused}}$ and $\mathcal{R}_{sem}$ if and only if $t \in L(N)$ and
  $\NFlreduceTree{\Gamma, t}{v}$ is a theorem over $\mathcal{R}_{sem}$.
\end{theorem}

Theorem~\ref{fused-correctness} states the soundness and completeness of the fused-semantics optimization.



\section{Evaluation}
\label{Se:Evaluation}

In this section, we evaluate the feasibility of our algorithm to solve \semgus problems 
through our implementation \ourtool.
Specifically, we investigate the following four issues:

\begin{description}
  \item \q1: We evaluate the effectiveness of \ourtool on \sygus benchmarks. 
  \item \q2: We evaluate the effectiveness of \ourtool on imperative program-synthesis problems.  
  \item \q3: We evaluate the effectiveness of the optimizations discussed in \S\ref{Se:encoding} and \S\ref{Se:Implementation}. 
  \item \q4: We evaluate the effectiveness of approximate semantics supplied to \ourtool.
    \cbstart
  \item \q5: We evaluate the effectiveness of \ourtool on regular-expression benchmarks.
    \cbend
\end{description}

Overall, our evaluation is tilted towards proving unrealizability compared to synthesizing programs.
This is because because there already exist many 
program synthesizers that incorporate multiple years of engineering effort~\cite{cvc4_orig, cvc4, sketch}; 
it is beyond the scope of this paper and \ourtool 
to directly compete with these synthesizers.



\subsection{Benchmarks}
\label{Se:evaluation-benchmarks}

We performed our evaluation using three sets of benchmarks.

\cbstart
\jinwoo{Review C: Where do these benchmarks come from?}
The first set consists of 132 unrealizable variants
of the 60 LIA (Linear Integer Arithmetic) benchmarks from the LIA \sygus competition track.
These benchmarks were generated by Hu et al.\ \cite{unreal} and have been used as benchmarks for unrealizability
in previous work~\cite{unreal, semilinear}.
These benchmarks originate from \sygus problems 
where the goal is to prove that a synthesized solution is optimal with respect to some metric~\cite{qsygus} 
(e.g., minimizing the number of If-Then-Else operators).
One can prove optimality by proving that a solution 
with a lower score is unrealizable.
\cbend

In each of the benchmarks, the grammar that specifies the search space is recursive, and hence
generates infinitely many LIA terms. These benchmarks are unrealizable because they contain grammars that
restrict how many times a certain operator (e.g., plus or if-then-else) can appear in the solution.
To see how effective \ourtool is as a synthesizer, we also test \ourtool on the 
60 original LIA \sygus benchmarks in \S\ref{Se:evaluation-sygus}.
These benchmarks have a completely unrestricted grammar, as opposed to the 132 unrealizable 
variants generated from them.

The second set consists of 289 imperative \semgus problems defined over various 
fragments of the imperative grammar $G_{impv}$.
Out of these, 36 benchmarks were created by hand from common imperative programming questions, 
such as synthesizing a Fibonacci function or swapping variables using bitwise-xor.
The remaining 253 benchmarks were derived by using the 30 benchmarks employed in a previous paper 
on synthesizing imperative programs via enumeration~\cite{SoO17} as a template.
Out of the 30 templates, we ignored 7 that contained division, on which Z3 would return an error, 
and derived 11 benchmarks from each of the 23 remaining templates for a total of 253 benchmarks.
The 23 base templates consist largely of two categories: those that compute a function over a 
range of numbers $1$ to $n$ using a loop (such as factorials or sums), and those that 
compute a function over an array, again using a loop to iterate 
(such as finding the maximum element of an array, or adding two arrays together).
To derive our benchmarks, we first instantiated \semgus with the problem specification and 
the unbounded grammar $G_{impv}$ with a restriction on the number of loops: the grammar 
in this case replicates the templates used to specify the search space from~\cite{SoO17}.
Then, various restrictions were imposed on the grammar, such as 
limiting the number of statements allowed, or limiting the kinds of expressions that can occur as the 
loop condition.
Out of the 11 benchmarks generated from each template, 
2 were designed to be realizable, and 9 to be unrealizable.
We developed our own set of benchmarks this way because the unrealizability of imperative programs
is a previously unstudied field.

\cbstart
The final set of benchmarks consists of the 25 regular-expression problems from textbooks on automata theory, 
where the specification is given as a set of positive and negative examples, averaging around 10 examples total~\cite{alpharegex}.
\cbend

Each benchmark was given 10 minutes to complete on a machine with a 2.6GHz Intel Xeon processor 
with 32GB of RAM, with version 4.8.9 of Z3 as the external CHC solver.
We note that the front-end processing step, to encode a \semgus problem into CHCs, took less than 3 minutes for all of our 
benchmarks and configurations combined; the 10-minute timeout was separate from the front-end processing step, and devoted entirely 
to CHC solving.

Table~\ref{table:summary_all} summarizes the numbers of solved benchmarks for various configurations of \ourtool we tested, 
as well as comparisons for \nay~\cite{semilinear}, ESolver~\cite{alur2017esolver}, CVC4~\cite{cvc4}, SIMPL~\cite{SoO17}, and 
AlphaRegex~\cite{alpharegex}.
\begin{table}
 \caption{Number of solved benchmarks for various configurations of \ourtool, 
  alongside results for \nay, ESolver, CVC4, and SIMPL. \xmark \ indicates cases where the tool is non-applicable. 
  SIMPL could only be evaluated on 23 realizable imperative benchmarks, 
  because SIMPL cannot accept a grammar.}
  \begin{tabular}{cc|rr|rr|r}
    & \multirow{2}{*}{\textbf{Solver}} & \multicolumn{2}{c|}{\sygus} & \multicolumn{2}{c|}{Imperative} & \multicolumn{1}{c}{Regex}\\
      & & \multicolumn{1}{c}{Realizable} & \multicolumn{1}{c|}{Unrealizable} & \multicolumn{1}{c}{Realizable} & 
      \multicolumn{1}{c|}{Unrealizable} & \multicolumn{1}{c}{Realizable}\\ \hline\hline
      & \multicolumn{1}{c|}{\nay} & \xmark & 70 & \xmark & \xmark & \xmark \\
      & \multicolumn{1}{c|}{ESolver} & 6 & \xmark & \xmark & \xmark & \xmark \\
      & \multicolumn{1}{c|}{CVC4} & 59 & \xmark & \xmark & \xmark & \xmark \\
      & \multicolumn{1}{c|}{SIMPL} & \xmark & \xmark & $23^*$ & \xmark & \xmark \\ 
      & \multicolumn{1}{c|}{AlphaRegex} & \xmark & \xmark & \xmark & \xmark & 25 \\ \hline
      \parbox[t]{1mm}{\multirow{7}{*}{\rotatebox[origin=c]{90}{\ourtool}}} \ \ & Total & 4 & 66 & 8 & 112 & 5 \\ \cdashline{2-7}
      & Fused Trees & 3 & 66 & 5 & 31 & \xmark \\ 
      & Fused Lists & 4 & 66 & 5 & 62 & \xmark \\ 
      & Fused Arrays & 2 & 64 & 6 & 91 & \xmark \\ 
      & Vectorized Arrays & 2 & 66 & 5 & 56 & \xmark \\ 
      & Individual  & 0 & 56 & 3 & 10 & 0 \\ \cdashline{2-7}
      & Abstract & \xmark & 18 & \xmark & 37 & \xmark \\ 
      & Underapproximate & \xmark & \xmark & 6 & \xmark & \xmark \\ \cdashline{2-7}
      & Enumerative & \xmark & \xmark & 2 & \xmark & 5 \\ \hline\hline
      & Total benchmarks  & 60 & 132 & 67 & 222 & 25 \\
  \end{tabular}
  \label{table:summary_all}
\end{table}

\subsection{\q1: Evaluating \ourtool on \sygus Benchmarks}
\label{Se:evaluation-sygus}

\cbstart
\jinwoo{Review C: CVC4 settings, Shepherd item 5: Enumerative solver}
In this section, we evaluate the effectiveness of \ourtool on \sygus benchmarks by 
comparing it against \nay, the state-of-the-art tool for proving unrealizablity for \sygus problems, 
and against the \sygus synthesizers CVC4, which was run using the default settings from \sygus-COMP 2019~\cite{cvc4sy}, and ESolver.
It is worth noting that both CVC4 and ESolver solve synthesis problems directly, while \ourtool relies on an external 
CEGIS loop to generate examples.
We did not test the enumerative \semgus solver described in \S\ref{Se:Implementation} for \sygus benchmarks, 
because ESolver already serves as a basic enumerator.
\cbend

Like \ourtool, \nay checks the (un)realizability of a \sygus problem when the specification is given as 
a set of examples; unlike \ourtool, however, \nay cannot synthesize a solution for realizable problems.
We used the set of 132 unrealizable \sygus benchmarks described in 
\S\ref{Se:evaluation-benchmarks} for evaluation.
We report that \ourtool or \nay solves a problem if it can solve a problem using any of its 
configurations---for \ourtool, this encompasses the three different term representations, 
as well as the different semantics that \semgus can be instantiated with.

We implemented a CEGIS algorithm for \ourtool, and compared it against the CEGIS algorithm of \nay.
Because \nay is incapable of synthesizing an answer to realizable a \sygus problem, 
the CEGIS loop of \nay relies on an external synthesizer ESolver to produce a term.
Because \ourtool and \nay rely on different methods to produce counterexamples,
their CEGIS iterations may differ.

	With the standard CEGIS algorithm, \name can prove 61/132 benchmarks unrealizable, while \nay can do so for 65/132.
	\nay also provides a modified ``random'' variant of CEGIS 
	that allows random examples to be added to the set of counterexamples throughout
	the CEGIS loop. Using this technique, \nay can prove unrealizability for an additional 5 benchmarks.
	We ran \ourtool on the same set of examples produced using this technique; it was able to
	prove unrealizability for the same 5 benchmarks as well, for a total of 66/132 benchmarks.
	There are 67 benchmarks where both solvers timed out---stuck at some iteration of the CEGIS loop. On 17 
	of them, \name can
	complete more iterations of the CEGIS loop than \nay (avg. 6.1 for \ourtool vs.\ 4.8 for \nay). 
  On 13 of them, \nay can progress further (avg. 2.2 for \ourtool vs.\ 3.1 for \nay). The two solvers were stuck at the same iteration on the rest of the benchmarks.
	Table~\ref{Ta:results} in Appendix~\ref{Appendix:results} presents a detailed comparison of the runtimes for each tool, 
	on benchmarks from the \textsc{LimitedIf} and \textsc{LimitedPlus} categories.

Next, we compared the abilities of \ourtool as a \sygus synthesizer 
on the 60 original LIA benchmarks upon which the unrealizable benchmarks were derived.
\ourtool solved 4/60 benchmarks.
\cbstart
\jinwoo{Review E: Addressed ESolver version}
This number is comparable to ESolver, which solved 6/60 benchmarks
\cbend
and was the winner of the first \sygus competition in 2014.
Moreover, \ourtool 
solved one benchmark that ESolver could not solve.
While \ourtool is not competitive with current \sygus solvers, such as CVC4~\cite{cvc4}, which solved 59/60 
benchmarks, the fact that its performance is already comparable with an early version of a \sygus solver 
is encouraging, and one might hope that more efficient algorithms for synthesizing solutions to \semgus problems are 
possible in the future.\footnote{
  We also note that the LIA \sygus benchmarks have an entirely free grammar and are single-invocation, which allows 
  CVC4 to use a specialized method involving quantifier elimination to synthesize programs.
}

\ourtool was efficient at proving unrealizability: 56 out of the 66 benchmarks solved were solved in under 
10 seconds, 
and \ourtool also has comparable runtimes with \nay.
For synthesizing programs, ESolver solved all solvable 
benchmarks in under two minutes each, while \ourtool required 6 minutes for two of the solved benchmarks.

\textbf{To answer \q1:} \ourtool \emph{is quite effective on unrealizable \sygus problems, to a degree that 
is comparable with \nay, and can also synthesize solutions for realizable \sygus problems}:
in particular, \ourtool is more general than previous tools as it can solve non-\sygus problems, and can 
produce two-sided answers to synthesis problem.

\subsection{\q2: Evaluating \ourtool on Imperative Synthesis Benchmarks}
\label{Se:evaluation-imperative}

In this section, we evaluate the effectiveness of \ourtool by seeing how well 
it can deal with imperative synthesis benchmarks.
We consider  \semgus-with-examples problems, as opposed to 
ordinary \semgus problems, due to the challenges of checking whether an imperative program 
satisfies a specification or not (which makes it difficult to implement a CEGIS loop).
In principle, one could implement a CEGIS loop using an external verifier.

Out of our 289 imperative benchmarks, a total of 67 were designed to be realizable, while 
the remaining 222 were designed to be unrealizable.
As shown in Table~\ref{table:summary_all}, \ourtool solved 8/67 realizable benchmarks, and 112/222 unrealizable 
benchmarks, for a total of 120 benchmarks solved.

\cbstart
\jinwoo{Shepherd item 3: Discussion of solved benchmarks}
Out of the 120 solved benchmarks, 10 benchmarks---2 realizable, and 8 unrealizable---were those with infinite syntactic search spaces and also 
contained the possibility of an infinite loop.
\ourtool also solved 15 benchmarks that did not contain loops, but nevertheless 
had infinite syntactic search spaces.
Overall, \ourtool had more success with proving unrealizability than synthesizing programs for both 
\sygus and imperative benchmarks, especially when a simple lemma proved sufficient for showing 
unrealizablity of the whole \semgus problem (such as the example given in \S\ref{Se:MotivatingExample}).

One benchmark proved unrealizable in a similar manner uses the grammar shown in Figure~\ref{fig:ex_grm1}, 
where the specification was to compute the sum of integers from $1$ to $x$ and store it in $y$.
Here, Z3 inferred the lemma that for the input example $[(x, y)] = [(2, 0)]$,
the values of $x$ and $y$ remain even regardless of the program 
being considered, because the constant $1$ is not present in the grammar---and thus cannot reach the correct output $y = 3$.

\begin{figure}[t!]
  \begin{subfigure}{0.47\textwidth}
    $$
    \begin{array}{llcl}
       Start & ::= & \Ewhile{B}{S} \\
       S & ::= & \Eassign{x}{E} \mid \Eassign{y}{E} \mid \Eseq{S}{S} \\
       E & ::= & x \mid y \mid E + E \mid  E - E \\
       B & ::= & E < E
    \end{array}
    $$
    \caption{A grammar for a benchmark where the goal is to synthesize a program adding integers from $1$ to $n$.} \label{fig:ex_grm1}
  \end{subfigure} \hspace{0.04\textwidth}
  \begin{subfigure}{0.47\textwidth}
    $$
    \begin{array}{llcl}
       Start & ::= & \Ewhile{B}{S} \\
       S & ::= & \Eassign{x}{E} \mid \Eseq{S}{S} \\
       E & ::= & x \mid y \mid E + E \mid 0 \mid 1 \mid -1 \\
       B & ::= & E < E
    \end{array}
    $$
    \caption{A grammar for a benchmark where the goal is to synthesize a program resulting in $x == y$.} \label{fig:ex_grm2}
  \end{subfigure}
  \caption{Grammars for example benchmarks that were proved unrealizable.}
\end{figure}

An example of a lemma that makes use of multiple input examples comes from a benchmark that uses the grammar shown in Figure~\ref{fig:ex_grm2},
where the goal was to set $x$ to be equal to $y$ by modifying $x$
in a while loop. This problem is unrealizable because the value of $x$ can increase
or decrease as the loop iterates, but cannot do both (based on what input is given).
When given the input examples $[(x, y)] = [(12, 20), (20, 12)]$, Z3 infers the lemma that ultimately when the loop terminates, 
``$x > 12 \ \mathsf{or} \ x < 20$''.
If $x > 12$, the second example is unsatisfiable, and if $x < 20$, the first example is unsatisfiable---thus 
Z3 is able to prove the whole synthesis problem as unrealizable.
\cbend

Another reason why \ourtool performs better on unrealizable problems is in
part due to how the generated CHCs are dealt with internally in Z3---as described above, 
Z3 proves unrealizability by discovering a lemma that conflicts with the specification.
For realizable problems, however, Z3 in the worst case must conduct a search over all possible 
concrete terms from a possibly infinite search space, in a process similar to generate-and-test.
The authors are unsure of whether Z3 is capable of discovering lemmas that can be used to prune the 
search space for realizable benchmarks; regardless of the answer, the results suggest that 
program execution expressed as CHC solving introduces overhead that
is large enough to make synthesis relatively more difficult compared to proving unrealizability.


\cbstart
\jinwoo{Shepherd item 5: Enumerative solver}
The baseline enumerative solver for \semgus succeeded in solving only 2 realizable benchmarks.
Interestingly, the runtimes for enumerative scripts (each of which verifies a candidate term) displayed a high amount of variance, 
with some scripts terminating in under a second, and others exceeding the timeout of 10 minutes.
There was no clear pattern behind the variation, except that terms containing loops tended to time out more compared to those  that did not contain loops.
This result is in line with the hypothesis about overhead in program execution expressed as CHC solving above---it 
seems that the internal algorithms of a CHC solver suffer for
the relatively simpler task of checking whether a proof tree is correct.
\cbend

In contrast, SIMPL~\cite{SoO17}, whose benchmarks we use in our evaluation of \ourtool, 
employs a strategy of performing static analysis in tandem with enumeration; 
SIMPL also employs heuristics that prefer smaller programs, and 
directly executes candidates to see if the specification is met.
This enumerative approach makes SIMPL perform better as a synthesizer:
SIMPL solves the full set of 23 realizable benchmarks upon which our benchmarks are based, 
while \ourtool can solve none.
However, SIMPL is incapable of 
proving unrealizability, because it is based on enumeration.
One also cannot express a syntactic search space in SIMPL outside of simple templates, which 
prevented us from running the rest of our realizable benchmarks on SIMPL.

\ourtool took less than 10 seconds to solve  82/120 benchmarks, and the other 12 benchmarks required more than a 
  minute to complete.
  Whether a benchmark contained an unbounded loop or an infinite search space seemed to have little correlation with the 
  runtimes: there were finite-search space benchmarks that took over a minute to complete, and benchmarks with both unbounded loops 
  and infinite search spaces that took less than a second to complete.
  This phenomenon suggests the importance of discovered lemmas in solving \semgus problems: given a strong lemma, 
  a \semgus problem can be solved quickly, even with infinite loops and search spaces.
  On the other hand, without a lemma, the problem can take a long time to solve even if the search space is finite.

\textbf{To answer \q2:} \ourtool \emph{is capable of solving \semgus problems with infinite search 
spaces and imperative semantics, especially if the given problem is unrealizable.}
Notably, \ourtool is the first tool that can prove unrealizablity for imperative synthesis problems.

\subsection{\q3: Evaluating Optimized Methods for Solving \semgus Problems}
\label{Se:evaluation-optimizations}

In this section, we compare the effectiveness of the various optimizations we described in 
$\S\ref{Se:encoding}$ and $\S\ref{Se:Implementation}$ for solving \semgus problems.
Specifically, we investigate the following two issues:

\begin{enumerate}
  \item We assess the effectiveness of the flattened term representation from \S\ref{Se:listing-sol}, by 
    comparing the performance of \ourtool configured to use trees, lists, and arrays as the term representation.
  \item We assess the effectiveness of vectorized and fused semantics, by comparing 
    the performance of \ourtool on (i) individual semantics, (ii) vectorized but non-fused semantics, and (iii) 
    vectorized and fused semantics.
\end{enumerate}

\subsubsection*{Effectiveness of Flattened Term Representations}

To evaluate the effectiveness of the three term representations, 
we supplied \semgus with vectorized semantics and 
enabled the fused-semantic-optimization, the configuration that yielded the overall best results 
in our evaluation.
In this section, we say that a particular term representation 
``solved'' a benchmark if Z3 was able to solve the CHCs produced by encoding the \semgus 
problem using the given term representation.

\cbstart
\jinwoo{Review E: Clarify Table 1}
The `Fused Trees', `Fused Lists', and `Fused Arrays' rows of Table~\ref{table:summary_all} summarize the results for the different 
\cbend
term representations: for \sygus benchmarks, all three representations were similar.
For imperative benchmarks, the array representation is clearly better compared to the list and 
tree representations: in particular, the tree representation could only solve one benchmark that 
the array representation could not solve, while the list representation solved 
a strict subset of the benchmarks solved by the array representation.

Figures~\ref{fig:lists_vs_trees} and~\ref{fig:lists_vs_arrays} compare the performances of
the list versus tree representations, 
and the list versus array representations, on imperative benchmarks solved by at 
least one of the representations.

\begin{figure}
  \begin{subfigure}{0.31\textwidth}
    \includegraphics[width=\linewidth]{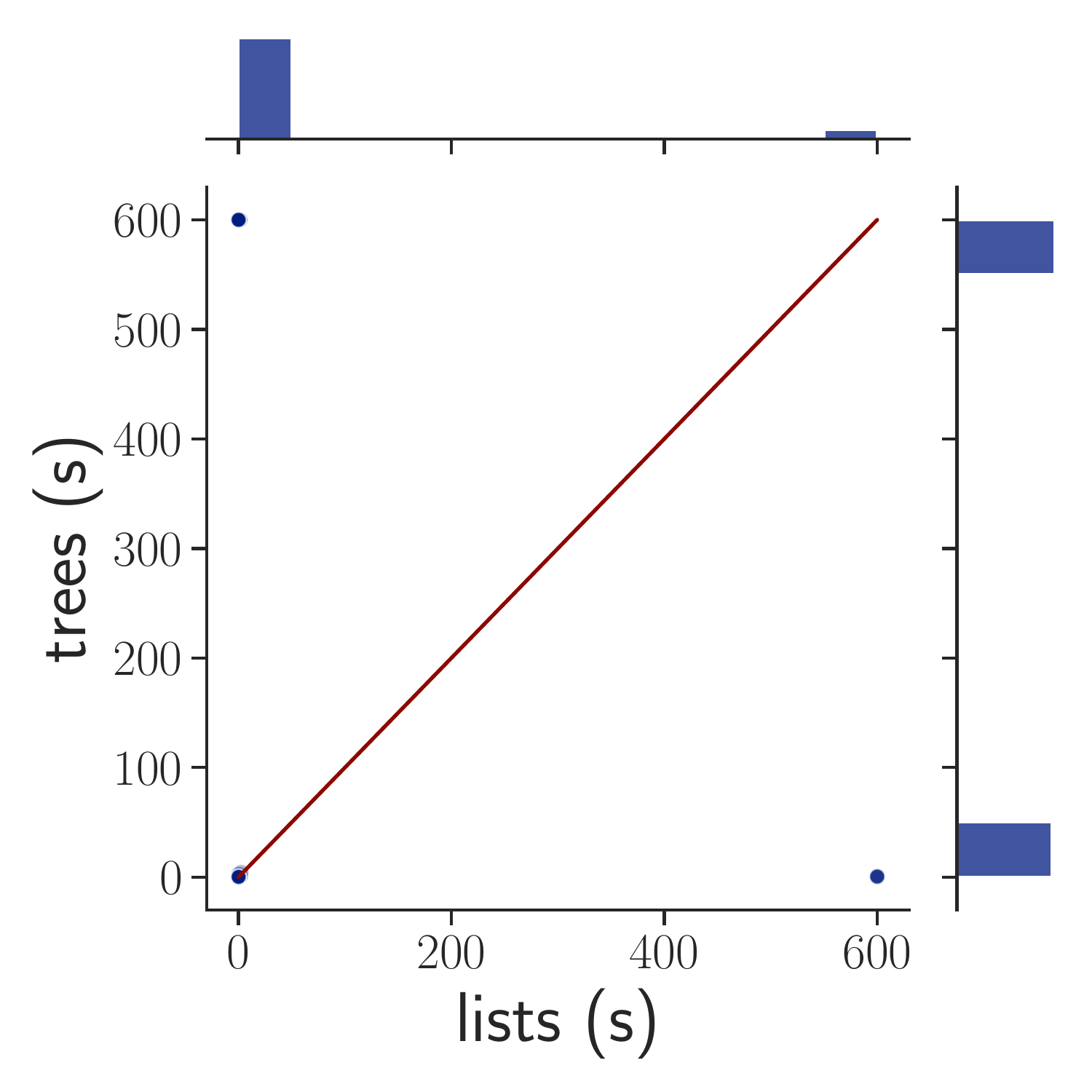}
    \caption{Runtime for list representations vs. tree representations, for 
    imperative benchmarks solved by at least one configuration.
    }
    \label{fig:lists_vs_trees}
  \end{subfigure} \hspace{0.015\textwidth}
  \begin{subfigure}{0.31\textwidth}
    \includegraphics[width=\linewidth]{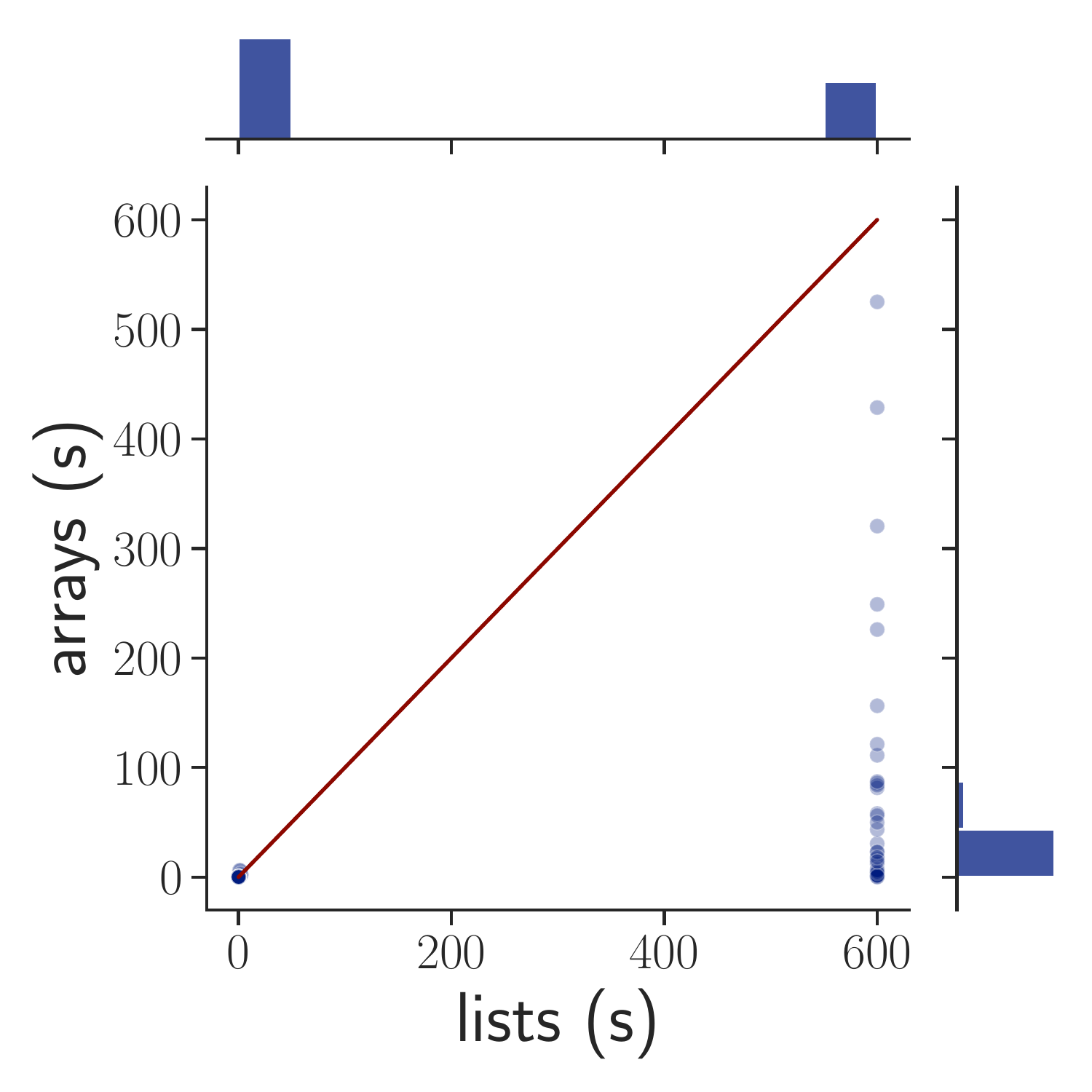}
    \caption{Runtime for list representations vs. array representations, for 
    imperative benchmarks solved by at least one configuration.}
    \label{fig:lists_vs_arrays}
  \end{subfigure} \hspace{0.015\textwidth}
  \begin{subfigure}{0.31\textwidth}
    \includegraphics[width=\linewidth]{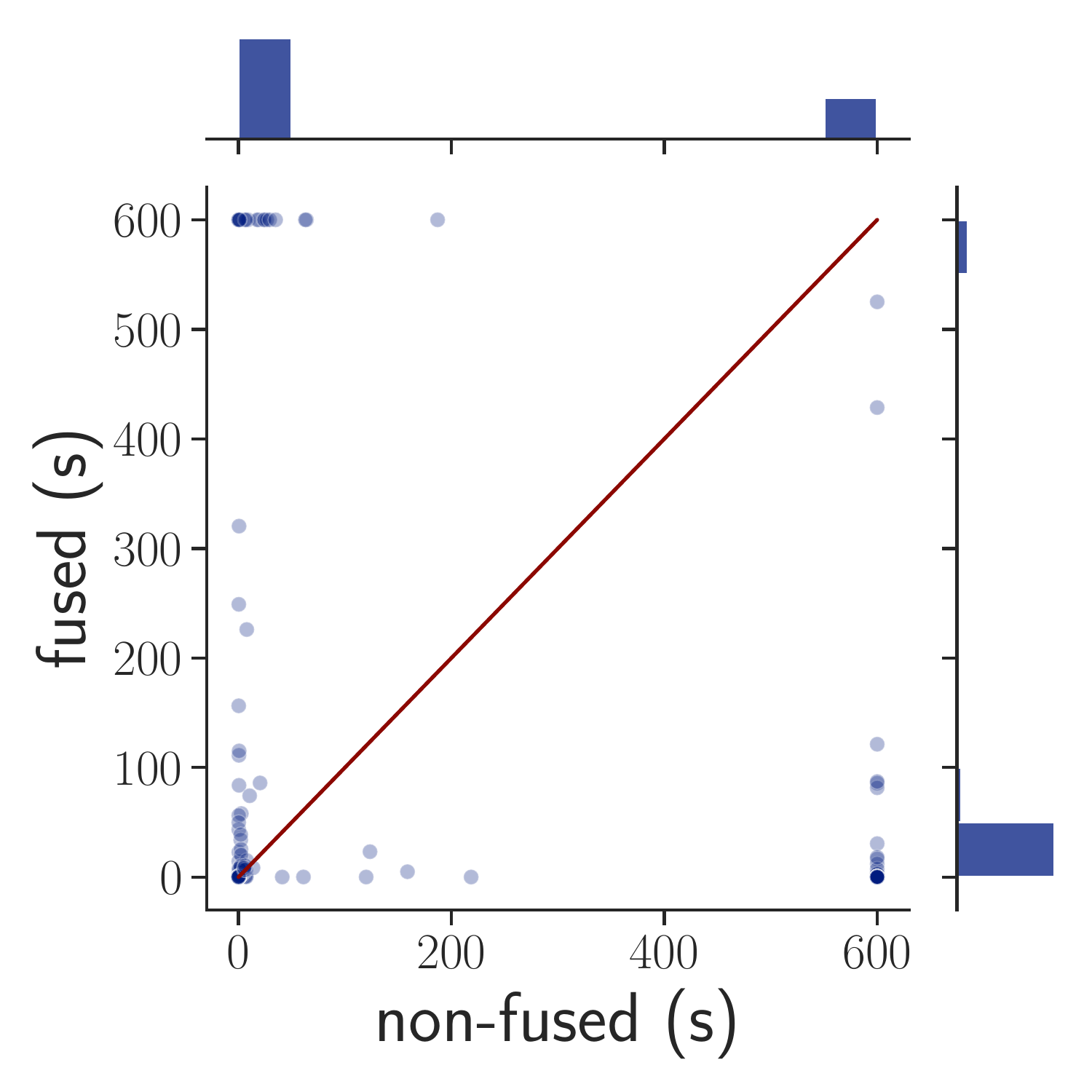}
    \caption{Runtime  for the non-fused vs. fused vectorized semantics, for benchmarks solved by at least one configuration.$\qquad \qquad$}
    \label{fig:vectorized_vs_fused}
  \end{subfigure}
  \caption{Runtime comparisons for various configurations of \ourtool.
  Bar graphs on the outer axes show the distribution of the data points.
  600 seconds indicates a timeout.}
  \label{fig:vs_graph}
\end{figure}

As mentioned in \S\ref{Se:listing-sol}, we suspect the differences between the different 
term representations is due to the fact that
support for algebraic datatypes in Z3 still remains relatively limited.\footnote{
  In our correspondence of the authors of Z3 and Spacer, 
  they mentioned that using inductive datatypes with the Horn Clause solver was highly experimental.
}
When the tree or list representations were used, Z3 terminated with the error ``stuck on a lemma'' far more 
often than when the array representation was used.

\textbf{To answer part} (1) \textbf{of \q3:} \emph{Flattened term representations are indeed effective, 
especially when using arrays to avoid the use of algebraic datatypes altogether.}

\subsubsection*{Effectiveness of the Vectorized and Fused Semantics}

To evaluate the effectiveness of the vectorized and fused semantics optimizations, 
we compared them against each other, and against individual semantics that are not 
vectorized nor fused (corresponding to the ``standard'' form of semantics mentioned in \S\ref{Se:Semantics}).
We note that while the vectorized semantics is actually a different semantics that 
\semgus can be supplied with, \ourtool can automatically vectorize the semantics for any subgrammar of 
$G_{impv}$ as an optimization; thus we treat it as an optimization for our evaluation.

\mypar{Individual semantics vs. Non-fused vectorized semantics}
Using a list representation of terms, the individual semantics could solve 56 unrealizable \sygus benchmarks, 
a strict subset of the 66 solved by the vectorized semantics.
For imperative benchmarks, the individual semantics could only solve 3 realizable and 10 unrealizable 
benchmarks, compared to 5 realizable and 62 unrealizable when using the vectorized semantics.

\mypar{Non-fused vectorized semantics vs. Fused vectorized semantics}
The ``Fused Arrays'' and ``Non-Fused Arrays'' rows of Table~\ref{table:summary_all} describe the number of 
solved benchmarks: again, performance for the \sygus benchmarks was similar.
For imperative benchmarks, the non-fused vectorized array semantics solved 5 realizable and 56 unrealizable 
benchmarks, compared to 6 realizable and 91 unrealizable for the fused vectorized array semantics.

Figure~\ref{fig:vectorized_vs_fused} compares the results of the non-fused vectorized versus the fused vectorized
semantics using array representations: while the fused semantics solve more benchmarks, the graph suggests that 
the fused semantics are not strictly better than the non-fused semantics.
In particular, there are 
11 unrealizable benchmarks that only the non-fused semantics can solve.


When using a list representation, the difference becomes less pronounced---the non-fused vectorized semantics 
can solve 4 realizable and 58 unrealizable benchmarks, compared to 
5 realizable and 62 unrealizable for fused vectorized list semantics.
We think the reason is again the limited support in Z3 for algebraic datatypes, 
which remains the main bottleneck when using list representations.

\textbf{To answer part} (2) \textbf{of \q3:} \emph{The fused vectorized semantics is effective as an optimization, especially for imperative
benchmarks, but there exist some benchmarks 
for which the non-fused vectorized semantics performs better.}
Both are consistently better than the individual semantics.

\subsection{\q4: Evaluating \semgus and \ourtool with Approximate Semantics}
\label{Se:evaluation-semantics}

We evaluated how well \ourtool performs when it is instantiated using an 
approximate semantics 
to produce one-sided answers to either synthesis or unrealizability.
In this section, we focus on identifying the number of \emph{new} benchmarks that 
an approximate semantics can solve compared to the standard version of the semantics.
This approach is motivated by the fact that the nature of an approximate semantics 
can sometimes change the kind of answer that can be obtained---for 
example, using an abstract semantics for an unrealizable problem might make it realizable---and 
thus a direct performance comparison makes little sense.

Both an abstract and underapproximating semantics were implemented using arrays as the 
term representation, with the fused-semantics and vectorizing optimizations, which 
displayed the best overall performance in \S\ref{Se:evaluation-optimizations}.
The new benchmarks solved are ones that the exact, fused and vectorized array semantics 
was unable to solve.

\mypar{Abstract semantics for unrealizability}
To test the capabilities of abstract semantics, we implemented five 
variants of the abstract domain 
$\mathbb{B}_3$ from $\S\ref{Se:abstract-semantics}$, where each domain tracked the 
first, second, third, fourth, and fifth bit of variables, respectively.
This choice was driven by the fact that most of the input examples for our imperative synthesis benchmarks 
were small, between 0 to 31.

The ``Abstract'' row of Table~\ref{table:summary_all} describes the number of benchmarks solved using 
the abstract semantics.
The abstract semantics did not make a difference for the \sygus benchmarks; all benchmarks that were solved 
by the abstract semantics were also solvable by the exact semantics.
However, for the imperative benchmarks, the abstract semantics was able to solve 
17 unrealizable benchmarks that the exact semantics could not solve.
Using abstract semantics yielded faster runtimes: the abstract semantics timed out 
for at most 15 benchmarks regardless of the variant of abstract domain, 
compared to over 200 for the exact semantics
(although realizability results in the abstract semantics have no meaning).

One reason that the abstract semantics failed to make a difference 
on the \sygus benchmarks could be that 
the \sygus benchmarks themselves used inputs with very small values, often between 0 and 7: 
the abstract semantics were able to prove some \sygus benchmarks as unrealizable using the lower 
bits, but nothing new.
In addition, the abstract domains that we used do not work well in the presence of addition, 
because carry bits often render a result to be $\top$.
All of our \sygus benchmarks contain addition, while some imperative benchmarks do not.

\mypar{Underapproximating semantics for program synthesis}
For the underapproximating semantics, we implemented the technique of 
bounding the number of loop-unrollings from 
\S\ref{Se:underapproximate-semantics}, and experimented with loop bounds of 10, 50, and 100.
We only compare the imperative benchmarks here, because the \sygus benchmarks do not contain loops.
The ``Underapproximate'' row of Table~\ref{table:summary_all} describes the number of benchmarks 
solved using this semantics.
The bound semantics was able to synthesize one more program compared to the non-bound semantics.
Interestingly, \ourtool succeeded in synthesizing the program (to compute the factorial function using 
a while loop) when the bound was set to 100, but not when the bound was 10 or 50.
The small difference in performance may be due to the fact that 
\ourtool generally performs worse as a synthesizer than a tool for proving unrealizability.
The results also tell us that synthesizing imperative programs is difficult, 
even without the presence of infinite loops: it could be because unrolled loops 
still pose a significant burden when trying to compute the semantics of an imperative 
program, especially because our approach must ultimately prove that a candidate term 
satisfies the specification using semantics encoded as CHCs 
(which is likely to be slower compared to direct execution).

\textbf{To answer \q4:} \emph{Abstract semantics allows} \ourtool 
\emph{to solve many more unrealizable 
\semgus problems compared to using only exact semantics.
The bound underapproximating semantics did allow} \ourtool \emph{to solve more realizable \semgus problems, 
but the improvement was small.}

\subsection{\q5: Evaluating \semgus and \ourtool on Regular Expressions}

\cbstart
\jinwoo{Shepherd item 4: regexes}
Finally, we evaluated the effectiveness of \ourtool on regular-expression
benchmarks.
The final column of Table~\ref{table:summary_all} displays the results: the CHC-solving approach failed to solve any of the 
benchmarks within 10 minutes,
while the enumerative approach described in \S\ref{Se:Implementation} 
solved 5 benchmarks.
This is an interesting case where the enumerative approach outperformed the CHC-solving approach.

The most likely hypothesis for this behavior is that the iterative nature of the $\mathsf{Star}$ rule
in Figure~\ref{fig:regex_final} (due to the unrolling parameter $k$),
as well as the behavior of strings themselves, 
mixes poorly with the overarching principle of finding a powerful lemma 
or invariant.
In addition, we did not impose restrictions on the grammar for regular expressions, which made all benchmarks 
realizable.

The enumerative solver, while succeeding in solving 5 benchmarks, displayed
behavior similar to what it displayed when solving imperative benchmarks:
the variance between termination times was high, without a clear pattern behind the variance.
In contrast, AlphaRegex~\cite{alpharegex}, from which we took our benchmarks, reports being able to solve all 
25 benchmarks; we believe this difference is due to the overhead in using a CHC solver instead of an automaton to check 
regexes, as well as the additional enumeration optimizations that AlphaRegex employs.

\textbf{To answer \q5:}
\textit{Regular-expression synthesis problems can be expressed in \ourtool.}
\ourtool could solve some of the benchmarks using enumeration,
but none of the benchmarks using the CHC-solving method.
The most likely explanation for this behavior is the difficulty of finding good lemmas for regex problems, upon which \ourtool---more specifically, 
Z3 and Spacer---heavily relies.
\cbend



\section{Related Work}
\label{Se:related}

\mypar{General Synthesis Frameworks}
Sketch~\cite{sketch} and Rosette~\cite{rosette} are both solver-aided languages, where one specifies a 
synthesis problem using a domain-specific language, which is 
translated into an SMT problem.
FlashMeta~\cite{flashmeta} is a synthesis framework that allows one to specify the 
semantics of operators in the language
using witness functions, which roughly correspond to the 
``inverse'' semantics of operators.
In these tools, the way synthesis problems are defined is directly 
tied with how they are solved: one needs to develop non-standard 
inverse semantics for FlashMeta, or phrase the synthesis problem 
within the language of Sketch or Rosette, which are requirements 
imposed by their particular synthesis algorithms.
Due to these reasons, these tools disallow defining (and therefore solving) synthesis problems involving infinite search spaces.

The first attempt to unify these frameworks into a logical one was  provided by \sygus~\cite{sygus}.
However, \sygus is not general enough as it cannot express synthesis problems over arbitrary
syntactic constructs that do not lie inside a decidable SMT theory.
\semgus, on the other hand, provides a \textit{logical} way to 
\textit{define} synthesis problems with \emph{custom} semantics. 
Moreover, the solving 
procedure for \semgus is \emph{motivated by the definition}, not the other way around.\footnote{
  One may argue that semantics expressed as a CHC is a restriction, 
  but as stated in \S\ref{Se:Introduction}, they are more of a formalization.
  One may also assume a different surface syntax, such as the format in 
  Equation~\eqref{eq:while-true}; a translation to CHCs is straightforward.
}
%

\mypar{Customizing Semantics in \sygus}
  \sygus allows one to provide semantics for user-defined terms, but the support is limited to 
   functions/operators
  that can be used in the grammar.
  Concretely, if we consider our formalization of semantics (Definition~\ref{Def:prod-semantics}), the degree of customization available 
  in \sygus is limited to customizing the constraint $\phi$ to a formula expressible in the background theory.\footnote{
    To be precise, this formula is further limited to a formula of the form $v_0 = f(v_1, \cdots, v_i)$, where $f$ is a non-recursive 
    function expressible in the background theory.
    Notably, this prevents expressing relations between $v_1, \cdots, v_i$, the results of nonterminals in the RHS.
  }
  This limitation prevents \sygus from expressing imperative program-synthesis problems; \semgus eases this restriction 
  by allowing one to replace $\phi$ with any first-order formula, as well as introduce new relations or arguments for the relations.

  \cbstart
  \jinwoo{Review C, E: comparison to CVC4 and its term encodings}
  CVC4 also has a setting where one can encode a grammar via algebraic datatypes, and build an interpretation of these 
  datatypes in terms of a background theory $T$ into an SMT solver~\cite{cvc4}.
  This feature is designed to extend the SMT-based synthesis approach of CVC4 towards grammars, and like \sygus, is restricted 
  to what one can express within the theory $T$.
  Semantics in \semgus and \ourtool are expressed using CHCs, which allow us to express a wider variety of semantics, and the 
  solving procedure is a proof search over CHCs as opposed to solving a universally quantified SMT formula (as in CVC4).
  \cbend

\mypar{Synthesis for Imperative Programs}
There have  been attempts at designing synthesizers specifically for 
imperative programs.
Existing tools require the user to provide templates that specify
most of the program~\cite{SrivastavaGF10,sketch}; the tools then
resort to various constraint-solving techniques to complete
missing parts of the template, which often do not contain loops~\cite{SrivastavaGF10}.

SIMPL~\cite{SoO17} can
synthesize imperative programs from input-output examples and 
a template that specifies most of the program.
SIMPL employs a simple \emph{enumeration-based} strategy, and uses
abstract interpretation to rule out templates that will not 
result in a solution.
Because SIMPL is based on enumeration, it performs well as a
synthesizer.
However, unlike \ourtool, SIMPL cannot restrict the terms allowed in a program and
it cannot establish that a problem is
unrealizable.

\mypar{Unrealizability}
\nope~\cite{unreal} and \nay~\cite{semilinear} are, to the best of our knowledge, 
the only two tools that can prove unrealizability for \sygus benchmarks
in which the grammar can generate infinitely many terms.
Because \nay consistently outperforms \nope, we only compare
against \nay  in our evaluation.
\ourtool can solve synthesis problems over any specified language, 
including imperative languages,
whereas both \nope and \nay can only solve \sygus problems.
One variant of \nay also uses Constrained Horn Clauses, 
which are used to encode the problem of solving a
set of equations that describes the sets of possible outputs of the program.
In \ourtool, the constraints are used
for describing both the syntax and the semantics of the programs in the search space.
Because of the syntactic constraints, \ourtool can extract the
synthesized program when the problem is realizable, which \nay is
unable to do.

There exist other tools that are capable of proving unrealizability in 
limited situations, such as CVC4~\cite{cvc4} or DryadSynth~\cite{sygus-comp}.
However, CVC4 can only prove unrealizability when the grammar is completely unrestricted~\cite{unreal, semilinear}.
DryadSynth does not accept a grammar as part of its specification;
\ourtool is the only tool that can perform synthesis and unrealizabilty for general \semgus problems.

\mypar{The Use of Semantics in Program Synthesis}
\emph{Synthesis using abstraction refinement} (SYNGAR)~\cite{POPL:WDS18} uses predicate
abstraction to prune the search space of a synthesis-from-examples problem.
SYNGAR builds a tree automaton representing all trees in the search space
that are correct with respect to an abstract semantics expressed using
predicate abstraction.
SYNGAR can be viewed as a special case of \semgus in which predicate abstraction is used to 
\emph{overapproximate} the semantics of terms in the programming language.
SYNGAR's approach is tied to the use of an \emph{abstract semantics} that can
be expressed using a finite abstract domain, whereas our 
approach extends to infinite domains.
In particular, with our approach, one can express the \emph{concrete
semantics} of a programming language.

FlashMeta~\cite{flashmeta}, is also a way of using 
semantics in program synthesis.

\mypar{The Use of Horn Clauses in Program Synthesis}
In Inductive Logic Programming (ILP) \cite{ML:Quinlan90,NGC:Muggleton91,Book:LD94},
given background knowledge, typically in the form of Horn Clauses,
the goal of ILP is to learn the defining formula for a logical relation
that agrees with a given classification of input examples.
Both ILP and our framework use Horn Clauses to specify background
knowledge---which for our algorithm consists of the syntax and
semantics of the target programming language.
However, the respective goals for the output answer are different:
(i) In ILP, the goal is to create a Horn-Clause program as the answer.
(ii) In our algorithm for \semgus, the goal is to create a program in the
    language that has been specified via the background knowledge.
Whether ILP techniques  can be adapted to
\semgus is left for future work.


\section{Conclusion and Future Work}
\label{Se:conclusion}

This paper develops \semgus, a new framework
for program synthesis that allows one to 
specify both the syntax and the semantics of a synthesis problem.
\semgus can be used for specifying synthesis problems over 
an imperative programming language; it also allows one to work with 
a variety of different semantics that may be better suited to 
solve a synthesis problem efficiently.
The paper also presents a general procedure for 
solving \semgus problems capable of both program synthesis and 
proving unrealizability, and an implementation \ourtool to solve \semgus problems.

\cbstart
\jinwoo{Review E: Future work}
\semgus opens many future directions of work.

\mypar{Inferring Lemmas for \semgus}
As mentioned in \S\ref{Se:MotivatingExample}, 
our procedure for solving \semgus problems 
relies on an external CHC solver to infer lemmas over 
sets of programs in the syntactic search space, using the semantics of terms.
While we have relied on an external solver (Z3) to perform this inference for us, 
it is also unclear to what degree CHC solvers 
are capable of discovering lemmas to prune a syntactic search space.
An algorithm to explicitly infer lemmas and
prune parts of the search space would be 
especially useful in enhancing our solving algorithm as a synthesizer,
allowing \ourtool to compete with state-of-the-art synthesizers as well.

\mypar{\semgus with Abstract Domains}
As discussed in \S\ref{Se:underapproximate-semantics}, \semgus provides a 
natural framework for discovering and composing abstract domains for program synthesis.
Coupled with the fact that lemmas play an important role in solving \semgus problems, 
a system for inferring lemmas efficiently using abstract domains, and vice versa, 
has the potential to make \semgus solvers more efficient.
We have already witnessed this in our paper: 
the abstract domain $\mathbb{B}_i$ in \S\ref{Se:abstract-semantics} was 
inspired by the lemma from \S\ref{Se:semgus-solve}, which in turn 
proved the benchmark from Figure~\ref{fig:ex_grm2} unrealizable.

\mypar{Specialized Solvers for \semgus}
One limitation of our approach was that the synthesis procedure was quite 
inefficient, especially for regular expressions.
Since \semgus now gives us a way to specify synthesis problems, is it possible to 
develop solvers that are specialized for specific \semgus scenarios?
An example of such an approach could be a better enumerative solver, 
which does not rely on CHC solving but employs a separate algorithm to check the 
validity of proof trees.
\cbend
%
%

\begin{acks}
Supported, in part, by a gift from Rajiv and Ritu Batra; 
by ONR under grants N00014-17-1-2889 and N00014-19-1-2318; 
by NSF under grants 1420866, 1763871, and 1750965; 
by the Wisconsin Alumni Research Foundation; 
by a Facebook fellowship;
and by a grant from the Korea Foundation for Advanced Studies.
We also thank our shepherd, Nadia Polikarpova, 
for making numerous suggestions to improve the paper.
Opinions, findings, or conclusions expressed 
in this publication are those of the authors, 
and do not necessarily reflect the views of the sponsoring agencies.
\end{acks}

\bibliographystyle{ACM-Reference-Format}
\bibliography{reference}

\appendix

\clearpage
\section{A Detailed Comparison of \ourtool and \nay on Unrealizable \sygus Benchmarks}
\label{Appendix:results}

Table~\ref{Ta:results} presents a detailed comparison of the runtimes for \ourtool and \nay
on the \sygus benchmarks from the \textsc{LimitedIf} and \textsc{LimitedPlus} categories.
Both solvers could easily solve all the benchmarks in the easier \textsc{LimitedConst} category with comparable running   times (less than 1 second on most benchmarks).

\setlength{\tabcolsep}{3pt}
  \begin{table}[H]
	\caption{
		Performance of \name and \nay on a selected set of benchmarks from the \textsc{LimitedIf} and \textsc{LimitedPlus} categories, 
    which are the benchmarks that best highlight the differences between the tools (other benchmarks were quickly solved by both tools).
    CEGIS indicates the performance when using the CEGIS algorithm (with an external synthesizer for $\namenay$)
		while Oracle Examples describes the performance of the two solvers on a predefined set of examples
		for which the problem is known to be unrealizable.
		The table shows the number of nonterminals ($|N|$), productions ($|\delta|$), and variables ($|V|$)
		in the problem grammar; the number of examples produced in the CEGIS loop ($|E|$);
		and the total running time of \name and \nay.		\xmark\ denotes a timeout.}
		\footnotesize
		\centering
		\begin{tabular}{cc|rrr|rr|rr|rr}
				&\multirow{3}{*}{\bf Problem} &\multicolumn{3}{c|}{\bf } &\multicolumn{4}{c|}{\bf CEGIS } & \multicolumn{2}{c}{\bf Oracle Examples }\\
		&	  & \multicolumn{3}{c|}{\bf Grammar}& \multicolumn{2}{c|}{\bf \namenay} &  \multicolumn{2}{c|}{\bf \name} &\multicolumn{1}{c}{\bf \namenay} & \multicolumn{1}{c}{\bf \name} \\
			 &   &  {\bf $|N|$}  &  {\bf $|\delta|$} & {\bf $|V|$} &\multirow{1}{*}{\bf $|E|$} &\multirow{1}{*}{\bf time (s)} & \multirow{1}{*}{\bf $|E|$} &\multirow{1}{*}{\bf time (s)} &{\bf time (s)} &{\bf time (s)}   \\
			    \hline
\parbox[t]{-11mm}{\multirow{11}{*}{\rotatebox[origin=c]{90}{\textsc{LimitedPlus}}}}	& guard1	&	7	&	24	&	3	&	2	&	1.56	&	2	&	5.37	&	0.41	&	0.2	\\
& guard2	&	9	&	34	&	3	&	3	&	19.27	&	3	&	26.12	&	21.51	&	1.12	\\
& guard3	&	11	&	41	&	3	&	1	&	1.3	&	1	&	4.46	&	0.07	&	2.88	\\
& guard4*	&	11	&	72	&	3	&	2	&	\xmark	&	2	&	\xmark	&	72.58	&	1.93	\\
& plane1	&	2	&	5	&	2	&	1	&	1.28	&	2	&	6.65	&	0.12	&	0.02	\\
& plane2	&	17	&	60	&	2	&	2	&	\xmark	&	5	&	\xmark	&	1.24	&	7.51	\\
& plane3	&	29	&	122	&	2	&	1	&	\xmark	&	4	&	\xmark	&	25.42	&	95.17	\\
& ite1*	&	7	&	2	&	3	&	2	&	2.6	&	3	&	7.7	&	2.07	&	0.13	\\
& ite2*	&	9	&	34	&	3	&	2	&	\xmark	&	7	&	\xmark	&	29.54	&	45.67	\\
&  sum\_2\_5	&	11	&	40	&	2	&	2	&	\xmark	&	4	&	\xmark	&	20.47	&	15.65	\\
&  search\_2	&	5	&	16	&	3	&	3	&	3.13	&	3	&	6.29	&	2	&	0.13	\\
&  search\_3	&	7	&	25	&	4	&	4	&	4.07	&	4	&	8.17	&	4.81	&	0.36	\\
\hdashline																			
\parbox[t]{1mm}{\multirow{13}{*}{\rotatebox[origin=c]{90}{\textsc{LimitedIf}}}}																			
&max2	&	1	&	5	&	2	&	4	&	1.42	&	4	&	1.48	&	0.18	&	1.48	\\
&max3	&	3	&	15	&	3	&	9	&	16.57	&	9	&	\xmark	&	9.67	&	\xmark	\\
&  sum\_2\_5	&	1	&	5	&	2	&	3	&	1.49	&	3	&	0.69	&	0.26	&	0.69	\\
&  sum\_2\_15	&	1	&	5	&	2	&	3	&	1.42	&	3	&	0.87	&	0.26	&	0.87	\\
&  sum\_3\_5	&	3	&	15	&	3	&	8	&	34.84	&	8	&	\xmark	&	29.85	&	\xmark	\\
&  sum\_3\_15	&	3	&	15	&	3	&	9	&	41.87	&	6	&	\xmark	&	31.03	&	\xmark	\\
&  search\_2	&	3	&	15	&	3	&	5	&	42.55	&	5	&	\xmark	&	29.92	&	\xmark	\\
& example1	&	3	&	10	&	2	&	3	&	1.41	&	3	&	1.12	&	0.16	&	6.54	\\
& guard1	&	1	&	6	&	2	&	4	&	1.38	&	4	&	0.43	&	0.14	&	0.46	\\
& guard2	&	1	&	6	&	2	&	4	&	1.54	&	4	&	0.49	&	0.24	&	0.23	\\
& guard3	&	1	&	6	&	2	&	4	&	1.47	&	4	&	0.46	&	0.25	&	0.85	\\
& guard4	&	1	&	6	&	2	&	4	&	1.37	&	4	&	0.58	&	0.13	&	0.21	\\
& ite1	&	3	&	15	&	3	&	8	&	3.59	&	8	&	\xmark	&	5.35	&	\xmark	\\
			

\end{tabular}
\label{Ta:results}
\end{table}

\clearpage
\section{Proofs}
\label{Appendix:Proofs}

This section presents proofs for the soundness and completeness of our approach (Theorem~\ref{tree-correctness}) 
and the listing term representation (Theorem~\ref{listing-correctness}), as well as the correctness of the 
vectorized semantics (Theorem~\ref{vectorized-correctness}).
It also proves that the domains described in \S\ref{Se:Semantics} are indeed abstract and underapproximating 
(Lemma~\ref{is-abstract} and Lemma~\ref{is-underapproximating}), and finally shows the correctness of the 
fused-semantics optimization (Theorem~\ref{fused-correctness}).

In the proofs, we denote rules using their implication forms: for example, the inference rule in 
Equation~\eqref{eq:semantics_ex} can be written as 
$\BFlreduceTree{\Gamma, b}{\Etrue} \wedge \SFlreduceTree{\Gamma, s}{\Gamma_1} \wedge \StartlreduceTree{\Gamma_1, \Ewhile{b}{s}}{\Gamma_{2}}
\implies \StartlreduceTree{\Gamma, \Ewhile{b}{s}}{\Gamma_{2}}$.

\begin{theorem*}[Soundness and Completeness, \ref{tree-correctness}]
  Consider a \semgus-with-examples problem $sem = (G, \forall x \in \examples.\psi(x, f(x)))$,
  equipped with semantic rules $\mathcal{R}_{sem}$, a specification set $\examples$, and the $\mathsf{Query}$ rule
  (Equation~\eqref{eq:general-query}).
  Let the CHC form of $G$ be $\mathcal{R}_{syn}$.
  Then, \textit{Realizable} is a theorem over $\mathcal{R}_{sem}$ and $\mathcal{R}_{syn}$
  if and only if the \semgus-with-examples problem $sem$ is realizable.
  Moreover, if \textit{Realizable}  is a theorem, then the value of $t$ in the
  $\mathsf{Query}$ rule satisfies $t \in L(G)$ and $\forall x \in \examples. \: \psi(x, \sem{t}(x))$.
\end{theorem*}

\begin{proof}
  We prove the theorem by a lemma proving the correctness of the syntax rules $\mathcal{R}_{syn}$.
  \begin{lemma}[Soundness and Completeness of $\mathcal{R}_{syn}$]
    \label{syn-lemma}
    Suppose $\mathcal{R}_{syn}$ is the CHC representation of a grammar $G$, as described in 
    \S\ref{Se:MotivatingExample} and \S\ref{Se:encoding}.
    Then for any nonterminal $N \in G$, and some term $t$, 
    $\NFrreduceTree{t}$ is a theorem of $\mathcal{R}_{syn}$ 
    if and only if $t \in L(N)$.
  \end{lemma}
  \begin{proof}
    By induction on the size of the term $t$ (the number of productions applied to produce $t$), denoted by $n$.
    \begin{itemize}
      \item Base case ($n = 1$):
        To see completeness, assume $t \in L(N)$.
        Then $t$ must be created by a production $P$ of the form $N \rightarrow t$,
        because $t$ is of size $1$, thus $P$ cannot contain any nonterminals in the RHS.
        The CHC representation for $P$ is $\Etrue \implies \NFrreduceTree{t}$; thus $\NFrreduceTree{t}$ holds.

        To see soundness, suppose $\NFrreduceTree{t}$ holds.
        The only rule in $\mathcal{R}_{syn}$ that can be used to derive the theorem $\NFrreduceTree{t}$ 
        is of the form $\Etrue \implies \NFrreduceTree{t}$, because $t$ is of size $1$ and all other rules 
        create terms with size at least $2$.
        That $\Etrue \implies \NFrreduceTree{t}$ exists in $\mathcal{R}_{syn}$ means there exists a 
        production $N \rightarrow t \in G$; thus $t \in L(N).$ 
      \item Induction hypothesis ($n \leq k$): Assume the Lemma holds for all $n \leq k$.
      \item Inductive Step ($n = k + 1$):
        To see completeness, assume $t \in L(N)$; because $n = k + 1$, the first production $P$ in the 
        leftmost derivation of $t$ must contain nonterminals in the RHS (otherwise, $t$ becomes size $1$).
        Without loss of generality, let $P = N \rightarrow op(N_1, \cdots, N_i)$. 
        The CHC representation of this rule is $\NFnrreduceTree{1}{t_1} \wedge \cdots \wedge \NFnrreduceTree{i}{t_i} \implies \NFrreduceTree{t}$, 
        where $t = op(t_1, \cdots t_i)$.
        The induction hypothesis holds on terms $t_1, \cdots, t_i$, thus 
        $\NFnrreduceTree{1}{t_1} \wedge \cdots \wedge \NFnrreduceTree{i}{t_i}$ holds and $\NFrreduceTree{t}$ holds as well.

        To prove soundness, assume $\NFrreduceTree{t}$ is a theorem of $\mathcal{R}_{syn}$.
        This implies that there exists a rule $\NFnrreduceTree{1}{t_1} \wedge \cdots \wedge \NFnrreduceTree{i}{t_i} \implies \NFrreduceTree{t}$
        in $\mathcal{R}_{syn}$; the existence of this rule implies the existence of the production 
        $N \rightarrow op(N_1, \cdots, N_i) \in G$. 
        Because the induction hypothesis holds for $\NFnrreduceTree{1}{t_1}, \cdots, \NFnrreduceTree{i}{t_i}$, 
        it follows that
        $t_1 \in N_1, \cdots, t_i \in N_i$, and thus $t \in N$.
    \end{itemize}
  \end{proof}

  The semantic rules are supplied by the user, and we assume they correctly encode the semantics of the \semgus problem.

	As the final step, recall the $\mathsf{Query}$ rule from Equation~\eqref{eq:general-query}.
  By Lemma~\ref{syn-lemma}, $\StartrreduceTree{t}$ holds if and only if $t$ is a valid term in $L(G)$;
  the semantic relations $\bigwedge_{\ex_i \in \examples} \StartlreduceTree{t, \ex_i}{o_i}$ hold if 
  and only if $t$ executed on the example set $\ex_i$ results in the outputs $o_i$, which satisfy the 
  specification.
  Thus the algorithm is both sound and (relatively) complete.
\end{proof}

\begin{theorem*}[Correctness of Listing Semantics, \ref{listing-correctness}]
  Let $\mathcal{R}^{\mathsf{List}}_{sem}$ be a set of semantic rules using a flattened representation
  of terms, created from the set of semantic rules $R_{sem}$.
  Then for any nonterminal $N$, $\NFlreduce{\Gamma, \aL_{in}}{v, \aL_{out}}$ is a theorem of
  $\mathcal{R}^{\mathsf{List}}_{sem}$ if and only if
  $\NFlreduceTree{\Gamma, t}{v}$ is a theorem of $\mathcal{R}_{sem}$, and $\aL_{in} = \aL_t \aL_{out}$
  (i.e., the concatenation of $\aL_t$ and $\aL_{out}$) where $\aL_t$ is the pre-order listing of a term
  $t \in L(N)$.
\end{theorem*}

\begin{proof} 
  The correctness of the syntax rules using flattened representations of terms can be proved in an identical 
  manner to Lemma~\ref{syn-lemma}.
  For the proof of correctness of the semantic rules, we proceed by induction on the height of the derivation tree 
  for $\NFlreduce{\Gamma, \aL_{in}}{v, \aL_{out}}$ (for completeness) and 
  $\NFlreduceTree{\Gamma, t}{v}$ (for soundness), denoted by $n$.

  \begin{itemize}
    \item Base case $(n = 1)$: 
      To see completeness when $n = 1$, notice $\NFlreduce{\Gamma, \aL_{in}}{v, \aL_{out}}$ must be proved by a 
      single rule without premises, namely $\Etrue \implies \NFlreduce{\Gamma, \aL_{in}}{v, \aL_{out}}$; 
      by Definition~\ref{Def:prod-semantics}, 
      the only productions that may create such a rule are of the form $N \rightarrow t \sballns{j}$, where $t$ is a leaf node.
      In this case, $\aL_{in} = \sballlrns{j}::\aL_{out}$ and the preorder listing of $t$ is $\sballlrns{j}$.
      Also note that $N \rightarrow t \sballns{j}$ must be equipped with the semantics 
      $\Etrue \implies \NFlreduceTree{\Gamma, t}{v}$; thus $\NFlreduceTree{\Gamma, t}{v}$ is a theorem of 
      $\mathcal{R}_{sem}$, and $t \in L(N)$.

      A similar proof works for soundness: that 
      $\NFlreduceTree{\Gamma, t}{v}$ is a theorem of $\mathcal{R}_{sem}$ for $t \in L(N)$ implies the existence 
      of a production $N \rightarrow t \sballns{j}$ 
      equipped with the semantics $\Etrue \implies \NFlreduce{\Gamma, \sballlrns{j}::\aL_{out}}{v, \aL_{out}}$;
      thus $\NFlreduce{\Gamma, \sballlrns{j}::\aL_{out}}{v, \aL_{out}}$ is a theorem of 
      $\mathcal{R}_{sem}^{\mathsf{List}}$.

    \item Induction hypothesis $(n \leq k)$: We assume the theorem holds for all $n \leq k$.

    \item Inductive step $(n = k + 1)$:
      Both completeness and soundness can be proved using a similar procedure to the base case.
      
      To see completeness, assume $\NFlreduce{\Gamma, \aL_{in}}{v, \aL_{out}}$ is a theorem of 
      $\mathcal{R}_{sem}^{\mathsf{List}}$ with a derivation tree of height $k + 1$.
      Since the height is greater than 1, the derivation tree must contain premises of the form 
      $\phi \wedge \NFnlreduce{1}{\Gamma, \aL_1}{v_1, \aL_2} \wedge \cdots \wedge \NFnlreduce{i}{\Gamma, \aL_i}{v_i, \aL_{out}} 
      \implies \NFlreduce{\Gamma, \sballlrns{n}::\aL_1}{v, \aL_{out}}$, 
      as stated in Equation~\eqref{eq:semantics-list}.
      Now notice the existence of such a rule in $\mathcal{R}_{sem}^{\mathsf{List}}$ implies the 
      existence of a production 
      $A_0 \rightarrow \sigma(A_1, \cdots, A_i) \sball{n}$, equipped with the semantics 
      $\phi \wedge \mathsf{sem}_{A_1}(\langle \Gamma_1, t_1 \rangle, v_1), \cdots, \mathsf{sem}_{A_i}(\langle \Gamma_i, t_i \rangle, v_i)
      \implies  \mathsf{sem}_{A_0}(\langle \Gamma, t \rangle, v_0)$;
      this in turn implies the rule 
      $\phi \wedge \mathsf{sem}_{A_1}(\langle \Gamma_1, t_1 \rangle, v_1), \cdots, \mathsf{sem}_{A_i}(\langle \Gamma_i, t_i \rangle, v_i)
      \implies  \mathsf{sem}_{A_0}(\langle \Gamma, t \rangle, v_0)$ is in $\mathcal{R}_{sem}$.

      Now note that, due to the induction hypothesis, the theorem holds for the premises 
      $\NFnlreduce{1}{\Gamma_1, \aL_1}{v_1, \aL_2}, \cdots, \NFnlreduce{i}{\Gamma, \aL_i}{v_i, \aL_{out}}$, 
      and thus $\mathsf{sem}_{A_0}(\langle \Gamma, t \rangle, v_0)$ is a theorem of $\mathcal{R}_{sem}$, $t \in L(N)$ and $\aL_{in} = \aL_t ++\aL_{out}$ (where $\aL_t$ is the pre-order listing of $t$).

      Soundness can be proved in an identical manner by reversing the flow of rules, and applying the 
      induction hypothesis at the last step.
  \end{itemize}

  We note that this proof assumes that the height of the derivation tree is finite, which 
  assumes that the given program terminates.
  This makes little difference, as our algorithm for program synthesis also requires that the 
  produced derivation tree is finite---as shown in \S\ref{Se:Evaluation}, this does not affect the practicality of our approach.
\end{proof}

Theorem~\ref{vectorized-correctness} states the correctness of the vectorized semantics from \S\ref{Se:vectorized-semantics}.

\begin{theorem*}[Correctness of Vectorized Semantics, \ref{vectorized-correctness}]
	Given a set of examples $\examples=[\Gamma_1,\ldots, \Gamma_n]$ and a term $t$,
	$\sem{t}_\examples([\Gamma_1,\ldots, \Gamma_n], [\Gamma_1',\ldots, \Gamma_n'])$
	if and only if for every $1\leq i\leq n$, we have
	$\sem{t}(\Gamma_i,\Gamma_i')$.
\end{theorem*}

\begin{proof}
  We proceed by structural induction on $G_{impv}$, from Figure~\ref{fig:grm}. 
  Note $\bot$ is a special state that ignores all computations, i.e., $\sem{t}(\bot, \bot)$ for any term $t$.
  \begin{itemize}
    \item Base case: Expressions ($BVExpr, IntExpr, BoolExpr$): 
      Take the expression $x$ as an example.
      It is clear that $\sem{t}_{\examples}([\Gamma_1, \ldots, \Gamma_n], [v_1, \ldots, v_n])$ 
      if and only if for every $1 \leq i \leq n$, $\sem{t}(\Gamma_i, v_i)$, because each $v_i$ 
      will be obtained using $\Gamma_i(x)$, remaining independent for all examples.
      Other base cases in the expressions category can be proved in a similar manner; 
      then one may apply structural induction on the expressions to prove the theorem for expressions.

    \item Assignments and Updates ($\Eassign{x}{E}, \Eassign{x}{C}, \Earrayupdate{arr}{E}{E}$): 
      Assignments and updates form the base case for other statements in $G_{impv}$.
      Take $\Eassign{x}{E}$ as an example.
      Then the rule $\mathsf{Assign_\examples}$ in Figure~\ref{fig:vector-sem-bottom} directly shows that 
      $\sem{t}_{\examples}([\Gamma_1, \ldots, \Gamma_n], [\Gamma'_1, \ldots, \Gamma'_i])$ 
      if and only if for every $1 \leq i \leq n$, $\sem{t}(\Gamma_i, \Gamma'_i)$, 
      as each $\Gamma_i$ is updated separately (as shown on the second premise).
      A similar argument works for the other cases.
    \item Sequential Composition ($\Eseq{s_1}{s_2}$):
      The vectorized semantic rule for $s_1; s_2$ is as following:
      $$
      \infer[\mathsf{Seq_\examples}]{\sem{\Eseq{s_1}{s_2}}_\examples(\VG, \VG_2)}
      {\sem{s_1}_\examples(\VG, \VG_1) \quad \sem{s_2}_\examples(\VG_1, \VG_2)}
      $$
      The case holds due to the induction hypothesis.
    \item Branch Statements ($\Eifthenelse{b}{s_1}{s_2}$):
      The vectorized semantic rule for $\Eifthenelse{b}{s_1}{s_2}$ is as following:
      $$
        \infer[\mathsf{SITE_\examples}]{\sem{\Eifthenelse{b}{s_1}{s_2}}(\VG, \Vjoin{\VG_1}{\VG_2})}
        {\sem{b}_\examples(\VG, \Vv_b) \quad \sem{s_1}_\examples(\Vproj{\VG}{\Vv_b}, \VG_1) \quad 
        \sem{s_2}_\examples(\Vproj{\VG}{\lnot \Vv_b}, \VG_2)
        }
      $$
      Where $\textsc{Proj}$ and $\textsc{Merge}$ are defined in Figure~\ref{fig:vector-sem-bottom}.
      By the induction hypothesis, the entries of $\VG_1$ are correct on those states which not have 
      been projected to $\bot$, i.e. $b$ evaluates to $\Etrue$; in other cases it is $\bot$.
      The same holds for $\VG_2$ for states where $b$ evaluates to $\Efalse$; in other cases it is $\bot$.
      Then by the definition of $\textsc{Merge}$, the theorem holds for if-then-else as well.

    \item Loops ($\Ewhile{b}{s}$):
      The rules $\mathsf{WTrue_\examples}$ and $\mathsf{WFalse_\examples}$ are the two semantic rules 
      associated with the term $\Ewhile{b}{s}$.
      To prove the case for while loops, we introduce an additional induction on the number of loop iterations, 
      denoted by $n$.
      \begin{itemize}
        \item Base case $(n = 0)$: In this case, there are 0 iterations and thus the tree consists of only 
          one application of $\mathsf{WFalse}_\examples$; the theorem holds as all conditions are $\Efalse$ 
          or $\bot$ (which implies the incoming state is $\bot$ as well), and the input states equal the output states 
          for all examples, for both the vectorized and standard semantics.
        \item Induction hypothesis $(n \leq k)$: Assume the theorem holds for all $n \leq k$.
        \item Inductive step $(n = k + 1)$:
          If $n \geq 1$, note that there must be at least one application of the $\mathsf{WTrue}_\examples$ rule 
          in the derivation tree.
          Consider the first of these applications, and note that the induction hypothesis holds for all three semantic premises:
          $\sem{b}_\examples(\VG, \Vv_b)$, $\eSreduce{s}{\Vproj{\VG}{\Vv_b}}{\VG_1}$ (structural induction), 
          and $\eSreduce{\Ewhile{b}{s}}{\VG_1}{\VG_2}$ (height of the derivation tree).
          Again, by the semantics of the $\textsc{merge}$ operator, the theorem holds for the inductive step as well:
          the $\textsc{Proj}$ and $\textsc{merge}$ operators send individual examples on which the loop should not iterate 
          (i.e., $v_{b_i}$ evaluates to $\Efalse$ for the $i$-th entry $\Gamma_i$) to $\bot$, and recovers them to their
          original values after iteration has finished for all examples in the vector.
      \end{itemize}
  \end{itemize}
\end{proof}

Lemma~\ref{is-abstract} states that the domain $\mathbb{B}_i$ from \S\ref{Se:abstract-semantics} satisfies 
our definition of an abstract semantics.

\begin{lemma}
  \label{is-abstract}
  The semantics $\sem{\cdot}^\#$ defined in \S\ref{Se:abstract-semantics} over the domain $\mathbb{B}_i$ is an abstract semantics for 
  $G_{impv}$, with respect to the standard semantics $\sem{\cdot}$ partly given in Figure~\ref{fig:vector-sem-top}, defined 
  over the integers $\mathbb{Z}$.
\end{lemma}

\begin{proof}
  The abstraction function $\alpha$ is given as $\alpha(v) = bit_i(v)$, where $bit_i(v)$ denotes the $i$th bit of $v$.
  The concretization function $\gamma$ can be given as $\gamma(\Etrue) = \{v \in \mathbb{Z} \mid bit_i(v) = \Etrue\}$, 
  $\gamma(\Efalse) = \{v \in \mathbb{Z} \mid bit_i(v) = \Efalse\}$, and $\gamma(\top) = \mathbb{Z}$; 
  $\alpha$ and $\gamma$ form a Galois connection, thus the semantics are abstract.
  The different abstract transformers presented in Equations~\eqref{eq:abstract-while-imprecise} ($\rWT_{\mathsf{Havoc}}^{\#}$)
  and~\eqref{eq:abstract-while-precise} ($\rWT_{\mathsf{Join}}^{\#}$) are
  both sound: clearly $\rWT_{\mathsf{Havoc}}^{\#}$ is sound as $\Gamma_r^\# = \top$ and thus all values in the state are
  concretized to $\mathbb{Z}$ through $\gamma$.
  Similarly, $\rWT_{\mathsf{Join}}^{\#}$ is sound as well as $\Gamma_r^\# = \top$, or, given $\sem{\Ewhile{b}{s}}(\Gamma, \Gamma_r)$ for
  the standard semantics, the elements of $\Gamma_r$ are guaranteed to be within the concretization of $\Gamma_r^\#$ by the definition
  of $\textsc{Join}$.
\end{proof}

Lemma~\ref{is-underapproximating} states that the underapproximating semantics from \S\ref{Se:underapproximate-semantics} and 
Example~\ref{ex:under-example} satisfies 
our definition of an underappoximating semantics.

\begin{lemma}
  \label{is-underapproximating}
  The semantics $\sem{\cdot}^\flat$ described in \S\ref{Se:underapproximate-semantics} and Example~\ref{ex:under-example} is an 
  underapproximating semantics with respect to the standard semantics $\sem{\cdot}$ partly given in Figure~\ref{fig:vector-sem-top}.
\end{lemma}

\begin{proof}
  $\sem{\cdot}^\flat$ is defined by adding the constraint $i > 0$ to the standard semantics.
  Because the semantic rules have identical structure, one can apply a simple structural induction to show that 
  a theorem $\sem{t}^\flat(\Gamma, v)$ can only be proved if $\sem{t}(\Gamma, v)$.
\end{proof}

Finally, Theorem~\ref{fused-correctness} states the soundness and completeness
of the fused semantics from \S\ref{Se:fused-semantics}.

\begin{theorem*}[Soundness and Completeness of Fused Semantics, \ref{fused-correctness}]
   Let $\mathcal{R}_{syn}^{\mathsf{fused}}$ denote a set of fused syntax rules created according to the fused-semantic optimization 
   described in \S\ref{Se:fused-semantics}, from a grammar $G$ equipped with a semantics $\mathcal{R}_{sem}$.
   Then for any nonterminal $N \in G$, $\NFrreduceFused{\langle \Gamma, t \rangle}{v}$ is a theorem over 
  $\mathcal{R}_{syn}^{\mathsf{fused}}$ and $\mathcal{R}_{sem}$ if and only if $t \in L(N)$ and 
  $\NFlreduceTree{\Gamma, t}{v}$ is a theorem over $\mathcal{R}_{sem}$.
\end{theorem*}

\begin{proof}
  Similar to the way we proved correctness of the listing semantics, we proceed by induction on the height of the derivation tree
  for $\NFlreduceTree{\Gamma, t}{v}$ (for completeness) and $\NFrreduceFused{\langle \Gamma, t \rangle}{v}$ (for soundness), denoted by $n$:
  the proof is essentially a merging of the proofs for Lemma~\ref{syn-lemma} and Theorem~\ref{listing-correctness}.

  \begin{itemize}
    \item Base case $(n = 1)$:
      To see completeness, assume that $t \in L(N)$ and $\NFlreduceTree{\Gamma, t}{v}$ is a theorem over $\mathcal{R}_{sem}$.
      Because the $n = 1$, $t$ is a leaf of size 1; and thus there must exist a production $N \rightarrow t$ equipped with the semantics 
      $\Etrue \implies \NFlreduceTree{\Gamma, t}{v}$.
      Such a production and its semantics is encoded using the fused-semantics optimization as 
      $\Etrue \implies \NFrreduceFused{\langle \Gamma, t \rangle}{v}$ in  $\mathcal{R}_{syn}^{\mathsf{fused}}$; 
      thus the theorem holds.

      To see soundness, assume that $\NFlreduceTree{\Gamma, t}{v}$ is a theorem of $\mathcal{R}_{syn}^{\mathsf{fused}}$.
      Because $n = 1$, the derivation tree for $\mathcal{R}_{syn}^{\mathsf{fused}}$ cannot contain any premises 
      and thus there must be a rule of the form $\Etrue \implies \NFrreduceFused{\langle \Gamma, t \rangle}{v}$ in $\mathcal{R}_{syn}^{\mathsf{fused}}$; 
      this implies the existence of a production $N \rightarrow t$, equipped with the semantics $\Etrue \implies \NFlreduceTree{\Gamma, t}{v}$, 
      in $G$.
      Thus $t \in L(N)$ and $\NFlreduceTree{\Gamma, t}{v}$ is a theorem over $\mathcal{R}_{sem}$; the theorem holds.
    \item Induction hypothesis $(n \leq k)$: We assume the theorem holds for all $n \leq k$.
    \item Inductive step $(n = k + 1)$:
      To see completeness, again assume that $t \in L(N)$ and $\NFlreduceTree{\Gamma, t}{v}$ is a theorem over $\mathcal{R}_{sem}$.
      As $n > 1$, there must be premises in the first rule applied for proving $\mathcal{R}_{sem}$; i.e., the rule must be of form 
      $\phi \wedge \NFnlreduceTree{1}{\Gamma, t_1}{v_1} \wedge \cdots \wedge \NFnlreduceTree{i}{\Gamma, t_i}{v_i} \implies  \NFlreduceTree{\Gamma, t}{v}$: 
      the existence of this rule in $\mathcal{R}_{sem}$ implies the existence of a production 
      $N \rightarrow op(N_1, \cdots, N_i) \in G$ equipped with the rule as a semantics 
      (for example, $N \rightarrow \Ewhile{B}{S}$ from Figure~\ref{fig:fused-stmt-rules} below).
      This rule in $G$ is encoded into the rule 
      $\phi \wedge \NFnrreduceFused{\langle \Gamma, t_1 \rangle}{v_1}{1} \wedge \cdots 
      \wedge \NFnrreduceFused{\langle \Gamma, t_i \rangle}{v_i}{i} \implies \NFrreduceFused{\langle \Gamma, t \rangle}{v}$ 
      in  $\mathcal{R}_{syn}^{\mathsf{fused}}$ (following the example of $N \rightarrow \Ewhile{B}{S}$, this would correspond to either 
      $syn^{\mathsf{fused}}_{S \rightarrow \Ewhile{B}{S}, \Etrue}$ or $syn^{\mathsf{fused}}_{S \rightarrow \Ewhile{B}{S}, \Efalse}$);
      note that every premise holds due to the induction hypothesis, and thus $\NFrreduceFused{\langle \Gamma, t \rangle}{v}$ is a theorem of 
      $\mathcal{R}_{syn}^{\mathsf{fused}}$ as well.
      Because the induction is on the \emph{height of the derivation tree}, the proof works for loops as well, given that the loop terminates.

			Soundness works in a similar procedure: supposing $\NFrreduceFused{\langle \Gamma, t \rangle}{v}$ is a theorem of 
      $\mathcal{R}_{syn}^{\mathsf{fused}}$, because $n > 1$, the first rule applied in deriving $\NFrreduceFused{\langle \Gamma, t \rangle}{v}$ must be of the 
      form  $\phi \wedge \NFnrreduceFused{\langle \Gamma, t_1 \rangle}{v_1}{1} \wedge \cdots
      \wedge \NFnrreduceFused{\langle \Gamma, t_i \rangle}{v_i}{i} \implies \NFrreduceFused{\langle \Gamma, t \rangle}{v}$.
      This in turn implies the existence of a production 
      $N \rightarrow op(N_1, \cdots, N_i) \in G$, equipped with the semantics  
      $\phi \wedge \NFnlreduceTree{1}{\Gamma, t_1}{v_1} \wedge \cdots \wedge \NFnlreduceTree{i}{\Gamma, t_i}{v_i} \implies  \NFlreduceTree{\Gamma, t}{v}$;
      by the induction hypothesis, $t \in L(N)$ and $ \NFlreduceTree{\Gamma, t}{v}$ is a theorem of $\mathcal{R}_{sem}$.
  \end{itemize}
\end{proof}

Figure~\ref{fig:fused-stmt-rules} displays fused syntax rules for the productions 
$S \rightarrow \Eassign{x}{E}$,
$S \rightarrow \Eseq{S}{S}$,
$S \rightarrow \Eifthenelse{B}{S}{S}$, $S \rightarrow \Ewhile{B}{S}$ in $G_{impv}$ 
as an example, generated by applying the fused-semantics optimization on standard semantics.

\begin{figure}[H]
  $$
   \infer[syn^{\mathsf{fused}}_{S \rightarrow \Eassign{x}{E}}]{\SFrreduceFused{\langle \Gamma, \Eassign{x}{e} \rangle}{\Gamma_1}}
   {\EFrreduceFused{\langle \Gamma, e \rangle}{v} \quad \Gamma_1 = \Gamma[x \mapsto v]}
  $$
  $$
   \infer[syn^{\mathsf{fused}}_{S \rightarrow \Eseq{S}{S}}]{ \SFrreduceFused{\langle \Gamma, \Eseq{s_1}{s_2} \rangle}{\Gamma_2} }
   {\SFrreduceFused{\langle \Gamma, s_1 \rangle}{\Gamma_1} \quad \SFrreduceFused{\langle \Gamma, s_2 \rangle}{\Gamma_2}}
  $$
  $$
   \infer[syn^{\mathsf{fused}}_{S \rightarrow \Eifthenelse{B}{S}{S}, \Etrue}]{ \SFrreduceFused{\langle \Gamma, \Eifthenelse{b}{s_1}{s_2} \rangle}{\Gamma_1} }
   {\BFrreduceFused{\langle \Gamma, b \rangle}{\Etrue} \quad \SFrreduceFused{\langle \Gamma, s_1 \rangle}{\Gamma_1}}
  $$
  $$
   \infer[syn^{\mathsf{fused}}_{S \rightarrow \Eifthenelse{B}{S}{S}, \Efalse}]{ \SFrreduceFused{\langle \Gamma, \Eifthenelse{b}{s_1}{s_2} \rangle}{\Gamma_1} }
   {\BFrreduceFused{\langle \Gamma, b \rangle}{\Efalse} \quad \SFrreduceFused{\langle \Gamma, s_2 \rangle}{\Gamma_1}}
  $$
  $$
   \infer[syn^{\mathsf{fused}}_{S \rightarrow \Ewhile{B}{S}, \Etrue}]{\SFrreduceFused{\langle \Gamma, \Ewhile{b}{s} \rangle}{\Gamma_2} }
   {\BFrreduceFused{\langle \Gamma, b \rangle}{\Etrue} \quad  \SFrreduceFused{\langle \Gamma, s \rangle}{\Gamma_1} \quad  \SFlreduceTree{\Gamma_1, \Ewhile{b}{s}}{\Gamma_2} }
  $$
  $$
   \infer[syn^{\mathsf{fused}}_{S \rightarrow \Ewhile{B}{S}, \Efalse}]{\SFrreduceFused{\langle \Gamma, \Ewhile{b}{s} \rangle}{\Gamma}}
   {\BFrreduceFused{\langle \Gamma, b \rangle}{\Efalse}}
  $$
  \caption{Example syntax rules using the fused semantics.}
  \label{fig:fused-stmt-rules}
\end{figure}

\end{document}